\documentclass{article}

\usepackage{microtype}
\usepackage{graphicx}
\usepackage{booktabs} % for professional 
\usepackage{hyperref}
\usepackage{subcaption}
\usepackage{graphicx}
\usepackage{subcaption}
\usepackage{soul}
\usepackage{bbding}
\usepackage{mathtools}
\usepackage{braket}
\usepackage[accepted]{icml2025}

\usepackage{amsmath}
\usepackage{amssymb}
\usepackage{mathtools}
\usepackage{amsthm}
\usepackage{physics}
\usepackage{amsmath}
\usepackage[capitalize,noabbrev]{cleveref}

\theoremstyle{plain}
\newtheorem{theorem}{Theorem}[section]
\newtheorem{proposition}[theorem]{Proposition}

\theoremstyle{definition}

\theoremstyle{remark}

\usepackage[textsize=tiny]{todonotes}
\usepackage{bm}

\icmltitlerunning{Quantum-Inspired Fidelity-based Divergence}
\begin{document}

\twocolumn[
% \icmltitle{Amplitude-Encoded Quantum Relative Entropy for Enhanced Divergence Measurement}
\icmltitle{Quantum-Inspired Fidelity-based Divergence}


\icmlsetsymbol{equal}{*}

\begin{icmlauthorlist}
\icmlauthor{Yifeng Peng}{1}
\icmlauthor{Dantong Li}{2}
\icmlauthor{Xinyi Li}{3}
\icmlauthor{Zhiding Liang}{4}
\icmlauthor{Yongshan Ding}{5}
\icmlauthor{Ying Wang}{6}

\end{icmlauthorlist}
\icmlaffiliation{1}{}
\icmlaffiliation{2}{}
\icmlaffiliation{3}{}
\icmlaffiliation{4}{}
\icmlaffiliation{5}{}
\icmlaffiliation{6}{}
% \icmlaffiliation{yyy}{Stevens Institute of Technology}
% \icmlaffiliation{comp}{Yale University}
% \icmlaffiliation{sch}{Rensselaer Polytechnic Institute}


\icmlkeywords{Machine Learning, ICML}

\vskip 0.3in
]




\begin{abstract}
% Kullback-Leibler (KL) divergence is a classic tool for measuring differences between probability distributions; however, it is sensitive to distribution support, causing instability in high-dimensional data. To address this, we propose a Quantum Fidelity-based Divergence (QIF) \dantong{1) do we want quantum in the name 2) divergence is almost always monotone, but flogf gets to extreme at around 0.4 and behaves like a quadratic function} method based on quantum information theory. We define a customized Quantum Relative Entropy (QRE) \dantong{quantum relative entropy is already well-defined and our flogf is not an entropy, which measures disorderness, justification needed} by calculating the fidelity of encoded pure quantum states and introducing a minimal regularization parameter to ensure numerical stability. Additionally, we derive upper bounds and continuous similarity measures for Quantum Machine Learning (QML). Considering current quantum hardware limitations, we develop a quantum SWAP Test approach for fidelity measurement on real quantum devices \dantong{we do not need nor should use a quantum computer to evaluate this quantity}. In addition, we proposed a new regularization method, QR-Drop, based on QIF. The experiment results show that compared with State-of-the-art (SOTA) methods, QR-Drop has the best anti-overfitting performance. Our code is available at \url{https://github.com/Yifengml}.
Kullback--Leibler (KL) divergence is a fundamental measure of the dissimilarity between two probability distributions, but it can become unstable in high-dimensional settings due to its sensitivity to mismatches in distributional support. To address robustness limitations, we propose a novel \textit{Quantum-Inspired Fidelity-based Divergence} (QIF), leveraging quantum information principles yet efficiently computable on classical hardware. Compared to KL divergence, QIF demonstrates improved numerical stability under partial or near-disjoint support conditions, thereby reducing the need for extensive regularization in specific scenarios. Moreover, QIF admits well-defined theoretical bounds and continuous similarity measures. Building on this, we introduce a novel regularization method, \textit{QR-Drop}, which utilizes QIF to improve generalization in machine learning models. Empirical results show that QR-Drop effectively mitigates overfitting and outperforms state-of-the-art methods. 
 % Kullback--Leibler (KL) divergence is a cornerstone measure for quantifying dissimilarity between two probability distributions; however, it can exhibit numerical instability in high-dimensional settings, particularly when the support of one distribution does not fully overlap with the other. In response, we propose a novel \emph{Quantum-Inspired Fidelity-based Divergence} (QIF). Drawing on principles from quantum information theory, QIF employs fidelity-based measurements for pure-state encodings. Crucially, we show that this fidelity can be computed efficiently on classical hardware, circumventing the need for resource-intensive quantum computations, which remain costly and error-prone under current technological constraints. Compared to KL divergence, QIF demonstrates improved numerical stability under partial or near-disjoint support conditions, thereby reducing the need for extensive regularization in certain scenarios. Moreover, QIF admits well-defined theoretical bounds and continuous similarity measures. Building on this, we introduce \emph{QR-Drop}, a novel regularization technique leveraging QIF to enhance generalization in machine learning tasks. Empirical evaluations indicate that QR-Drop effectively mitigates overfitting and achieves performance gains over established baselines, highlighting the practical advantages of quantum-inspired fidelity. Our implementation is publicly available at: \url{https://anonymous.4open.science/r/QIF-QR-Drop-6835}. 

\end{abstract}

\begin{figure}[ht!]

\begin{center}
\centerline{\includegraphics[width=0.5\textwidth]{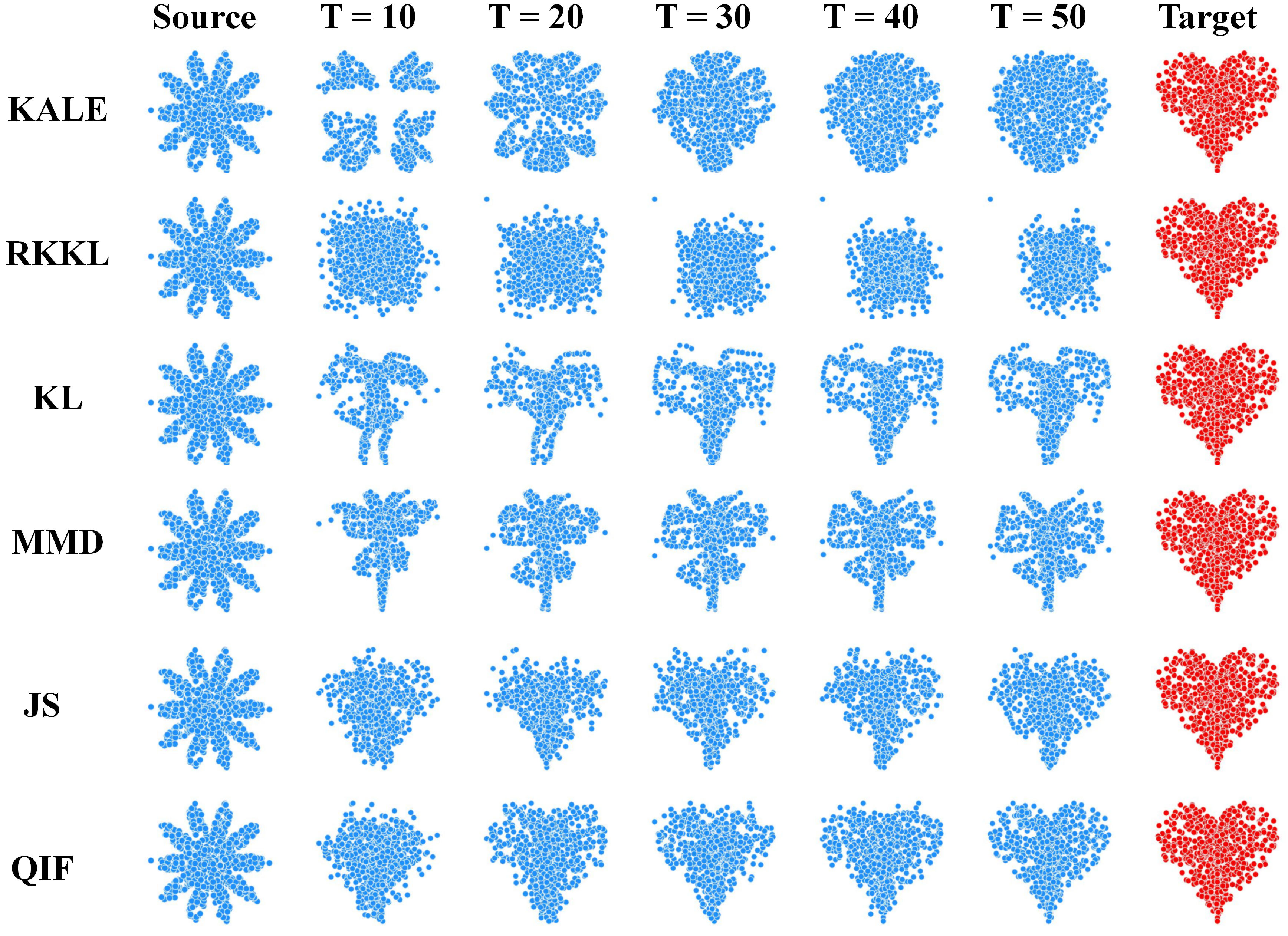}}
\caption{Distribution evolution of different divergence methods (KALE, RKKL, KL, MMD, JS, QIF) and sinkhorn distance during optimization. The blue distribution in the beginning stage noted as $\textcolor{blue}{\bm{\ast}}$ is the initial distribution of the input of all algorithms, the red heart distribution noted as $\textcolor{red}{\bm{\heartsuit}}$ is the target distribution, and T is the number of iterations with $\sigma = 0.3$, learning rate $ = 0.01$ and $1000$ sample points.}
\label{rkkl}

\end{center}

\end{figure}

\section{Introduction}
\label{Introduction}
A key challenge in machine learning is approximating an unknown probability distribution $Q$. In Bayesian inference \cite{gelman1995bayesian}, for example, one often seeks to approximate posterior distributions over model parameters for predictive tasks. This has led to the development of various techniques, including parametric methods such as variational inference~\cite{blei2017variational}, non-parametric approaches like Markov Chain Monte Carlo (MCMC)~\cite{roberts2004general}, and more recently, particle-based optimization~\cite{liu2016stein,korba2021kernel}. The Kullback-Leibler (KL) divergence \cite{kullback1951information} is a fundamental measure for assessing the quality of approximations in probability distribution learning. It quantifies the difference between two probability distributions, \( P \) and \( Q \), defined as
\begin{equation}
D_{\text{KL}}(P \| Q) = \sum_{i} P(i) \log \frac{P(i)}{Q(i)}.
\end{equation}
KL divergence has been extensively utilized in statistics, information theory, and machine learning due to its clear theoretical foundations and interpretability \cite{moreno2003kullback, van2014renyi} with complexity $\mathcal{O}(n)$. However, as data dimensions and volumes grow, traditional KL divergence encounters several mathematical and computational limitations that hinder its applicability in large-scale and high-dimensional settings, such as the "mode collapse" problem \cite{kingma2013auto} and cases requiring symmetric measures or robustness to distributional support differences \cite{mackay2003information}. One primary limitation of KL divergence is its sensitivity to the support of the distributions. Specifically, when the target distribution \( Q \) assigns zero probability to events that have non-zero probability under \( P \), the KL divergence becomes infinite. This sensitivity poses significant challenges in practical scenarios, especially in high-dimensional spaces where exact support overlap is rare \cite{gretton2012kernel, dauphin2014identifying}, leading to unstable and often uninterpretable results.

The Jensen-Shannon divergence \cite{menendez1997jensen} (JS divergence) is a symmetrized and smoothed alternative to the Kullback-Leibler (KL) divergence, defined for probability distributions \(P\) and \(Q\) as
\begin{equation}
\text{JS divergence}(P \parallel Q) = \frac{1}{2} \text{KL}(P \parallel M) + \frac{1}{2} \text{KL}(Q \parallel M),
\end{equation}
where \( M = \frac{1}{2}(P + Q) \). Mathematically, JS divergence offers the advantages of symmetry, boundedness between 0 and \(\log 2\), and finite values even when \(P\) and \(Q\) have non-overlapping supports, making it more stable and suitable for measuring mutual similarity. Additionally, by incorporating the average distribution \(M\), JS divergence helps mitigate issues like mode collapse with the same level of complexity compared to KL with $\mathcal{O}(n)$. However, JS divergence may face challenges in capturing nuanced differences between highly complex distributions due to its smoothing nature.

Recently, the Kernel Kullback-Leibler (KKL) divergence \cite{kkl} was introduced as an extension of the classical Kullback-Leibler (KL) divergence to kernel methods. For kernel matrices \( K_P \) and \( K_Q \) representing quantum states \( P \) and \( Q \), the KKL divergence can be defined as
\begin{equation}
\text{KKL}(P \parallel Q) = \text{Tr}\left(K_P \log K_P - K_P \log K_Q\right),
\end{equation}
where \(\text{Tr}(\cdot)\) denotes the trace operator. This formulation leverages the properties of semi-positive definite matrices, similar to density matrices in quantum information. It utilizes high-dimensional feature Hilbert spaces \(\mathcal{H}\) to enhance pattern recognition and information measurement capabilities, analogous to quantum computing operations \(\mathcal{H}\). However, KKL faces significant challenges, including numerical instability, overfitting, high sensitivity to noise, and computational complexity, which typically scales as \(\mathcal{O}(n^3)\) for matrix operations. The KKL requires that the support of distribution \(p\) be included in the support of distribution \(q\), i.e., \(\text{supp}(p) \subseteq \text{supp}(q)\). If this condition is violated, the KKL value becomes \(+\infty\), making it unsuitable for comparing distributions with non-overlapping supports. The root cause of this issue lies in the dependency of KKL on the covariance operator \(\Sigma_q\), where the kernel condition \(\text{Ker}(\Sigma_q) \subseteq \text{Ker}(\Sigma_p)\) must hold. When \(p\) and \(q\) have disjoint supports, this condition breaks down.

To address the limitations of KKL, a regularized version, \textit{RKKL}, was introduced \cite{regularizedkkl}:
\begin{equation}
\text{RKKL}_\alpha(P \| Q) = \text{KKL}\bigl(P \,\bigl\|\, (1 - \alpha)Q + \alpha P\bigr),
\end{equation}
where \(\alpha \in (0,1)\) is a regularization parameter that ensures \(\text{RKKL}_\alpha(P \| Q)\) remains finite even when \(P\) and \(Q\) have disjoint supports. The added mixture term prevents singularities, making RKKL well-defined for all distributions.

By introducing the mixture distribution $(1 - \alpha)Q + \alpha P$, RKKL guarantees non-zero probability density across all support regions, thus remaining well-defined even for disjoint $P$ and $Q$. This regularization removes the singularities in KKL while retaining the divergence property $\,\text{RKKL}_{\alpha}(P\|\!Q)=0\iff P=Q\,$. As $\alpha \to 0$, RKKL converges to the original KKL for overlapping supports. However, computing RKKL and its gradients requires $\mathcal{O}((n+m)^3)$ operations, and tuning $\alpha$ demands careful cross-validation to balance smoothness and accuracy, which can be computationally burdensome.

While alternative measures address some of these challenges, they introduce smoothing effects, regularization dependencies, or computational inefficiencies. We propose the Quantum-Inspired Fidelity-based Divergence (QIF). Unlike prior quantum-based divergences, QIF remains computationally efficient on classical hardware while maintaining robustness in high-dimensional learning tasks. Grounded in quantum information theory, QIF is defined as \begin{equation}
\text{QIF}(P \parallel Q) = - F(P, Q) \log F(P, Q),
\end{equation}
where \( P \) and \( Q \) are two classical distributions and $F$ indicates the fidelity. Different with quantum relative entropy (QRE) \cite{nielsen2010quantum}, \(\mathrm{Tr}\bigl[\rho(\log\rho - \log\sigma)\bigr]\), QIF uses the simple entropy-like expression (e.g., Shannon entropy, Von Neumann entropy \cite{shannon1948mathematical,von2013mathematische}) \(-x\log x\) with \(x=F(\rho_P,\rho_Q)\in [0,1]\). Consequently, it lies in \(\bigl[0,e^{-1}\bigr]\) and does not diverge even for disjoint supports, obviating the need for regularization parameters with the complexity \(\mathcal{O}(n)\) derived in Sec.\ref{sec:qif_classical}. As illustrated in Fig. \ref{rkkl}, leveraging fidelity-based computation mitigates computational bottlenecks and stabilizes performance, even for vanishing or disjoint supports. Inspired by quantum information, QIF is fully implementable on classical hardware discussed in Sec.\ref{sec:qif_classical}. Unlike traditional quantum divergences that require matrix exponentiation or density matrix calculations, QIF has less complexity.

% Although grounded in quantum fidelity, QIF is fully implementable on classical hardware. Unlike traditional quantum divergences that require matrix exponentiation or density matrix calculations, QIF’s fidelity-based formulation reduces computational overhead to $\mathcal{O}(n)$, making it scalable for machine learning applications.

% In deep learning, the KL divergence is used for measuring differences between probability distributions to prevent overfitting and enhance generalization \cite{wager2013dropout, baldi2013understanding}. A standard regularization method, dropout, randomly deactivates neurons to reduce reliance on specific features and improve robustness \cite{srivastava2014dropout}. However, the randomness in traditional dropout can slow and destabilize model convergence. To mitigate this, R-Drop \cite{wu2021rdrop} enhances training efficiency and performance by applying distinct dropout masks to inputs and minimizing their KL divergence to ensure output consistency. However, computing KL divergence in high-dimensional models suffers from numerical instability with extreme distribution differences. To address these challenges, we proposed the QKL-based QR-Drop. 

 In deep learning, KL divergence is commonly used to align output distributions for regularization and to mitigate overfitting \cite{wager2013dropout,baldi2013understanding}. Dropout \cite{srivastava2014dropout} randomly deactivates neurons to improve robustness, yet its randomness can destabilize convergence. R-Drop \cite{wu2021rdrop} addresses this by applying two dropout masks and minimizing their KL divergence, enforcing output consistency. However, in high-dimensional scenarios with heavily skewed or partially disjoint distributions, KL becomes numerically unstable. So we propose \emph{QR-Drop}, which replaces KL with QIF. Since QIF \(\in [0,e^{-1}]\) does not require additional parameters to handle disjoint supports, it avoids divergence pitfalls while maintaining a smooth and well-defined gradient. QR-Drop can achieve more stable training, particularly in large-scale tasks where distributions differ significantly or exhibit sparse modes.

\vspace{-2mm}
\section{Related Work}
\label{relatedwork}

\textbf{Distribution Approximating} Kullback--Leibler (KL) divergence is a standard tool for quantifying discrepancies between probability distributions, yet it is often computationally expensive and numerically unstable in high dimensions \cite{kingma2013auto,mackay2003information,gretton2012kernel,dauphin2014identifying}. Variations such as the normalized KL divergence \cite{bonnici2020kullback} and the cumulative residual KL (CRKL) distance \cite{laguna2019entropic} mitigate issues of unboundedness and address nodal structures in quantum distributions. Still, their reliance on quantum representations can limit applicability in continuous settings. Alternative measures, including Maximum Mean Discrepancy (MMD) \cite{arbel2019maximummmd} and Kernel Adaptive Least-squares Estimator (KALE) \cite{glaser2021kale}, reduce computational overhead and offer robust theoretical guarantees, albeit with drawbacks related to kernel choice, optimization complexity, and sensitivity to local minima. In addition, Quantum Relative Entropy \cite{nielsen2010quantum,wilde2013quantum} extends KL divergence into the quantum domain, potentially enhancing precision \cite{vedral2002role}, but matrix logarithms, trace operations, and challenges of quantum state preparation currently limit its widespread application.

\textbf{Dropout Regularization}
Dropout-based regularization has been pivotal in mitigating overfitting and improving generalization in deep neural networks \cite{srivastava2014dropout,baldi2013understanding}. Early efforts introduced standard dropout \cite{srivastava2014dropout}, randomly zeroing neuron activations to prevent co-adaptation. Subsequent studies advanced dropout’s versatility: Adaptive dropout \cite{ba2013adaptive} dynamically tuned dropout rates, Bayesian-inspired dropout \cite{gal2016theoretically} improved uncertainty estimation (later combined with reinforcement learning \cite{gal2017concrete}), and attention-driven dropout layers \cite{choe2019attention} enhanced object localization. Most recently, learned dropout patterns \cite{pham2021autodropout} further boosted robustness and generalization. Structured dropout methods refined the granularity of deactivation. Spatial dropout \cite{tompson2015efficientspatialdrop} zeroed entire feature maps or channels, preserving spatial relationships crucial for CNNs. DropBlock \cite{ghiasi2018dropblock} eliminated contiguous regions to thwart trivial bypasses of dropout. Meanwhile, R-Drop \cite{wu2021rdrop} enforced prediction consistency across multiple forward passes with dropout masks, reducing variance and stabilizing training.
\vspace{-2mm}
\section{Proposed QIF Divergence Algorithm}
% In this section, we discussed the specific implementation of the proposed QIF algorithm, especially weather quantum hardware or quantum simulation tools should be implemented for QIF.

In this section, we formally define the proposed QIF divergence algorithm, outlining its theoretical foundation and computational properties. We describe how QIF leverages fidelity-based measures while maintaining efficiency in classical computation. 
\vspace{-10pt}
\subsection{Quantum Derivation of QIF}
\label{QFQIF3}
% \hl{alternative subtitle: Quantum Foundation of QIF}
\textbf{Embedding Distributions into Quantum States} Let $p = (p_1, p_2, \ldots, p_d)$ and $q = (q_1, q_2, \ldots, q_d)$ be two classical probability distributions over $d$ outcomes, satisfying
$ p_i \ge 0, \quad \sum_{i=1}^d p_i = 1, q_j \ge 0$, $\sum_{j=1}^d q_j = 1$. To embed these distributions into quantum states (using amplitude encoding), we define:
$\ket{p} = \sum_{i=1}^d \sqrt{p_i} \,\ket{i}$, $ \ket{q} = \sum_{j=1}^d \sqrt{q_j} \,\ket{j}$. Each of these states is a unit vector in a Hilbert space $\mathcal{H}$ of dimension $d$, assuming $\sum_{i=1}^d p_i = 1 \quad \Rightarrow \quad \sum_{i=1}^d \abs{\sqrt{p_i}}^2 = 1$, and similarly for $q$.

\textbf{Representing the States as Density Matrices} In density-matrix form, the pure states $\ket{p}$ and $\ket{q}$ correspond to $\rho = \ket{p}\bra{p}, \sigma = \ket{q}\bra{q}
$. Since these are pure states, $\rho$ and $\sigma$ each have rank one and trace 1.

\textbf{Fidelity Between Two Pure States} For two density matrices $\rho$ and $\sigma$, the fidelity \cite{nielsen2010quantum} is defined as
\begin{equation}
\label{eqcalssicalfedility}
    F(\rho, \sigma) = \Bigl(\mathrm{Tr}\bigl(\sqrt{\sqrt{\rho}\,\sigma\,\sqrt{\rho}}\,\bigr)\Bigr)^2.
\end{equation}
However, when both $\rho$ and $\sigma$ are \textcolor{red}{pure states}, it becomes:
\begin{equation}
    F(\rho, \sigma) 
    = \abs{\langle p \mid q \rangle}^2.
\end{equation}
Substituting $\ket{p}$ and $\ket{q}$ into the inner product, we have $ \langle p \mid q \rangle = \sum_{i=1}^d \sqrt{p_i}\,\sqrt{q_i}$. Hence,
\begin{equation}
\label{eq:classicalqif}
    F(\rho, \sigma) = \biggl(\sum_{i=1}^d \sqrt{p_i}\,\sqrt{q_i}\biggr)^2.
\end{equation}

\textbf{Definition} We use an entropy-like definition as
\begin{equation}
\label{eq9}
    \mathrm{QIF}(\rho, \sigma)
    = -\,F(\rho, \sigma)\,\log\bigl(F(\rho, \sigma)\bigr).
\end{equation}
This expression is not the standard QRE (which is typically 
$\mathrm{Tr}(\rho(\log\rho - \log\sigma))$); rather, it is a simplified measure. Here, $F(\rho, \sigma)$ is the fidelity between the two pure states $\rho$ and $\sigma$:
\begin{equation}
\label{eq10}
    \mathrm{QIF}(\rho, \sigma)
    = -\,\abs{\langle p \mid q \rangle}^2
      \,\log\bigl(\abs{\langle p \mid q \rangle}^2\bigr).
\end{equation}
\textbf{Numerical Stability} In practice (such as in the provided code), to avoid taking the logarithm of zero, a small numerical constant $\varepsilon$ (e.g., $10^{-13}$) is added inside the clamp function:
\begin{equation}
\label{numericalsigma}
F(\rho,\sigma)
\;\leftarrow\;
\max\bigl(F(\rho,\sigma),\,\varepsilon\bigr).
\end{equation}
This ensures numerical stability when the fidelity is extremely small, different from the regularization parameter \(\alpha \in (0, 1)\) in KKL. We only need to make $\varepsilon$ small enough without adjusting. Detailed complexity analysis of amplitude embedding QIF. On IBM hardware, a SWAP test can measure fidelity in \(\mathcal{O}\bigl(n \,\frac{1}{\epsilon^2}\bigr)\), offering a practical alternative to full tomography.

\vspace{-10pt}
\subsection{Classical Reformulation of QIF}
% \hl{Alternative subtitle: Classical Reformulation of QIF: A Quantum-Inspired Model} 
\label{sec:qif_classical}
Quantum state preparation by amplitude embedding suffers from high costs for $2^n$ vector \cite{ciliberto2018quantum}. However, for \textcolor{red}{pure states}, when these density matrices are diagonal in the same basis, the quantum fidelity in Eq. \ref{eqcalssicalfedility} simplifies exactly to Eq. \ref{eq:classicalqif} as
\begin{equation}
\label{eq:bhattacharyya_coefficient}
F(P, Q) =  \Bigl(\sum_{i=1}^{d} \sqrt{p_i}\,\sqrt{q_i}\Bigr)^2,
\end{equation}
whose $\sum_{i=1}^d \sqrt{p_i}\,\sqrt{q_i}$ is closely related to the Bhattacharyya coefficient~\cite{bhattacharyya1943measure}, and it will form the core of our definition, which can be performed totally by classical computing. So, we can reformulate QIF by simplifying Eq. \ref{eqcalssicalfedility} in Sec.\ref{QFQIF3}.

Quantum computing offers huge advantages over traditional neural networks \cite{ding2020square,peng2024qsco, li2024efficient}. However, Noisy Intermediate-Scale Quantum (NISQ) \cite{preskill2018quantum} devices still suffer from noise and errors, creating challenges for pure quantum algorithms \cite{wang2022quantumnas, liang2022variational, liang2024napa}. According to Eq. \ref{eq:bhattacharyya_coefficient}, we avoid amplitude encoding of the probability distributions $p$ and $p$ to quantum state $\ket{p}$ and $\ket{q}$. 
\begin{figure}[ht!]
\begin{center}
\centerline{\includegraphics[width=0.45\textwidth]{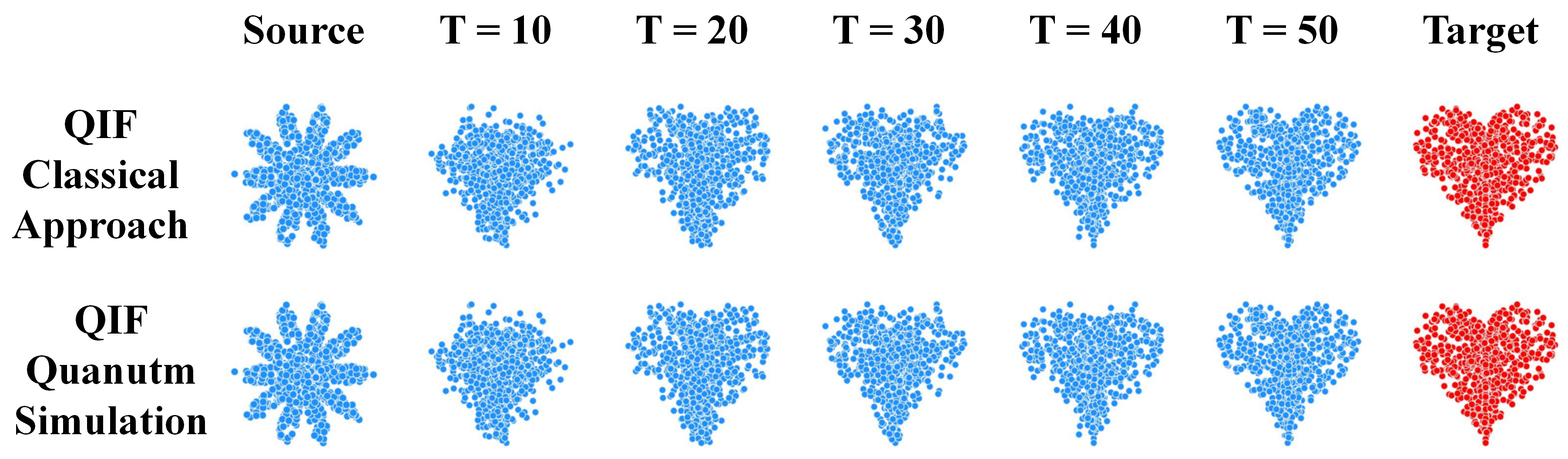}}
\caption{Comparison of quantum approach QIF by PennyLane \cite{pennylane} simulation and quantum-inspired QIF with the classical approach under the same setting.}
\label{comqifquantum}
\end{center}
\vspace{-25pt}
\end{figure}

As a result, fidelity reduces to a simple overlap between two probability vectors. Hence, the entire QIF measure, originally conceived in a quantum context, can be computed on classical hardware with no need for exponential dimension scaling. Then, we don't need to perform QIF using quantum hardware like measuring quantum fidelity by SWAP test or tomography and keep the same result in Fig. \ref{comqifquantum}. 
 
\textbf{Relation to Bhattacharyya Distance} Our construction leverages the Bhattacharyya coefficient (Eq.~\eqref{eq:bhattacharyya_coefficient}), which is well-known in information theory and statistics for measuring distributional overlap. One may also define the \emph{Bhattacharyya distance} \cite{bhattacharyya1943measure} as 
\begin{equation}
D_B(p,q)
\;=\;
-\ln \!\Bigl(\textstyle\sum_{i}\sqrt{p_i\,q_i}\Bigr).
\end{equation}
\textbf{Complexity} All operations (summing and logarithm) are \(O(n)\) with classical calculation. No \(2^n\)-dimensional vector or amplitude encoding step is required \cite{mottonen2004transformation}. 
\vspace{-10pt}

\section{Comparative Analysis of QIF and QRE}
In quantum information theory, QRE is regarded as the natural extension of the classical Kullback--Leibler divergence to the quantum realm~\cite{nielsen2010quantum}. While QRE retains several desirable properties of its classical counterpart, its application to pure quantum states reveals inherent limitations. Specifically, QRE becomes either zero or infinite for pure states, failing to provide a nuanced measure of similarity. In contrast, quantum fidelity offers a continuous similarity measure between quantum states \cite{horodecki2009quantum, bengtsson2017geometry}. Moreover, for pure states, quantum fidelity could be computed through a purely classical formulation, thus not only preserving fidelity's advantages but also reducing the computational overhead compared to a full quantum simulation.
\vspace{-10pt}
\subsection{Characteristic Analysis of QRE and QIF}

QRE between two density matrices $\rho$ and $\sigma$ is defined as:
\begin{equation}
S(\rho \| \sigma) = \text{Tr}(\rho (\log \rho - \log \sigma)),
\end{equation}
provided the support of $\rho$ is contained within $\sigma$.

\begin{proposition}
For pure states $\rho$ and $\sigma$, Quantum Relative Entropy is given by:
\begin{equation}
S(\rho \| \sigma) = 
\begin{cases}
0 & \text{if } \ket{\psi} = \ket{\phi} \\
+\infty & \text{otherwise}
\end{cases}
\end{equation}
\end{proposition}

\begin{proof}
Since $\rho = \ket{\psi}\bra{\psi}$ and $\sigma = \ket{\phi}\bra{\phi}$ are pure states, their supports are one-dimensional. If $\ket{\psi} = \ket{\phi}$, then $\rho = \sigma$, and clearly $S(\rho \| \sigma) = 0$.

If $\ket{\psi} \neq \ket{\phi}$, the support of $\rho$ is not contained within the support of $\sigma$, leading to $S(\rho \| \sigma) = +\infty$.
\end{proof}

The binary nature of QRE for pure states renders it ineffective as a similarity measure. It fails to capture the degree of similarity when $\ket{\psi} \neq \ket{\phi}$, as it only distinguishes between identical and non-identical states.

\textbf{Complexity Analysis} QRE for pure states primarily involves the computation of the inner product, resulting in a computational complexity of $O(n)$, and mixed states is $O(n^3)$. 

\begin{theorem}
For two probability distributions $p$ and $q$, QIF divergence provides a continuous similarity measure. 
\end{theorem}

\begin{proof}
Since $p$ and $q$ are probability distributions, all $p_i,q_i \ge 0$ and $\sum_i p_i = \sum_i q_i = 1$. We define the classical fidelity as Eq. \ref{eq:bhattacharyya_coefficient},
which lies in the interval $[0,\,1]$. Consequently,
\[
D_{\mathrm{QIF}}(p\|q) 
\;=\;
-\,F(p,q)\,\log\!\bigl(F(p,q)\bigr)
\;\in\;
[0,e^{-1}],
\]
because the function $-\,x \log(x)$ attains its maximum $1/e$ on $x \in [0,1]$. Thus $D_{\mathrm{QIF}}$ provides a \emph{nuanced}, \emph{continuous} measure of similarity between $p$ and $q$.
\end{proof}

% \textbf{Continuity of Quantum Relative Entropy} For two quantum states (density matrices) \(\rho\) and \(\sigma\) acting on a Hilbert space \(\mathcal{H}\), the Quantum Relative Entropy is defined as:
% \begin{equation}
% S(\rho \| \sigma) = \text{Tr}(\rho (\log \rho - \log \sigma)),
% \end{equation}
% where \(\log\) denotes the matrix logarithm.

\begin{proposition}
Quantum Relative Entropy \( S(\rho \| \sigma) \) is \textbf{lower semi-continuous} in the pair \((\rho, \sigma)\).
\end{proposition}
Lower semi-continuity implies that for any convergent sequence \((\rho_n, \sigma_n) \to (\rho, \sigma)\), the following holds:
\begin{equation}
    \liminf_{n \to \infty} S(\rho_n \| \sigma_n) \geq S(\rho \|\sigma).
\end{equation}
This ensures that the relative entropy does not "drop" in the limit but may "increase".

\begin{proposition}
Quantum Relative Entropy \( S(\rho \| \sigma) \) is \textbf{continuous} at \((\rho, \sigma)\) provided that the support of \(\rho\) is contained within the support of \(\sigma\), i.e., \(\text{supp}(\rho) \subseteq \text{supp}(\sigma)\).
\end{proposition}

\begin{proposition}
Quantum Relative Entropy \( S(\rho \| \sigma) \) is \textbf{not continuous} in general. Specifically, if \(\text{supp}(\rho)\) is not contained in \(\text{supp}(\sigma)\), then \( S(\rho \| \sigma) \) can exhibit discontinuities.
\end{proposition}
\vspace{-10pt}

\subsection{Advantages of QIF for Machine Learning Task}
To demonstrate the advantages of QIF in offering a bounded and continuous measure of similarity, we present the following proofs and explanations.

\begin{proposition}
The QIF Divergence \( D_{\text{QIF}}(p\| q) = -F(p, q) \log F(p, q) \) is bounded within \([0,e^{-1}]\).
\end{proposition}

\begin{theorem}
The boundedness and continuity of \( D_{\text{QIF}}(p \|q) \) enable the effective use of gradient-based optimization techniques in machine learning algorithms.
\end{theorem}

Gradient-based methods require the loss function to be smooth and differentiable. Since $D_{\text{QIF}}(p \| q)=-\, F\,\log(F)$ is bounded within \(\bigl[0,1/e\bigr]\), its overall scale remains finite, mitigating issues such as unbounded losses. Moreover, $\frac{d}{dF}\Bigl(-\,F\,\log(F)\Bigr)=-\log(F)\,-\,1$,which is well-defined and continuous for \(\,F\in(0,1]\). However, we note that as \(F \to 0\), \(-\log(F)\) becomes large; in practice, one usually assumes \(\rho\) and \(\sigma\) are not perfectly orthogonal so that \(F\) does not vanish. In numerical implementations, small regularization or smoothing strategies can also handle the case \(F \approx 0\) as $\sigma$ in Eq. \ref{numericalsigma}.

By contrast, QRE can diverge to \(+\infty\), making it less convenient for gradient-based methods. The finite and smooth character of \(D_{\text{QIF}}\) ensures stable gradient calculations, enabling effective navigation of the parameter space without encountering discontinuities or undefined regions.

\section{Understanding $F\log F$ in Divergence}

A common way to measure the distance between two states is \(\mathcal{L}_{\text{simple}} = 1 - F\),  which ensures zero loss when \(\ket{\psi} = \ket{\phi}\) (\(F = 1\)). However, for \(F \ll 1\), the gradient \(\frac{\partial}{\partial F}(1 - F) = -1\) can be too weak to drive effective updates. To address this, one can employ an entropy-like mapping \(\mathcal{L}_{\text{ent}} = -F \log(F)\), which provides stronger gradients — especially in the mid to low range \(F\), while remaining finite as \(F \to 0^+\). In QIF, \(F \in [0,1]\) typically denotes fidelity, so the negative sign ensures a nonnegative, bounded measure that vanishes at \(F=1\). By contrast, when \(F\) refers to a classical divergence such as KL or JS, the situation differs: KL can exceed 1, making \(-F \log(F)\) negative or even unbounded below, which is problematic for interpreting it as a loss, whereas \(F \log(F)\) effectively amplifies gradients for large KL values. Likewise, although JS lies in \([0,\log 2]\subset(0,1)\), it need not vanish at \(F=1\), so \(F \log(F)\) can straightforwardly reshape gradients without requiring a nonnegative loss. Consequently, while QIF and KL/JS may involve a form of \(\pm F\log(F)\), they operate over distinct ranges and objectives. Then, we choose $F \log F$ for KL and JS.

For discrete distributions, \(D_{\mathrm{JS}}(P \| Q)\) as $F$ is bounded above by \(\log(2)\). More precisely,
\begin{equation}
    0 \;\le\; D_{\mathrm{JS}}(P \| Q) \;\le\; \log(2).
\end{equation}
Consider a function \( G(F) = F \log F \). We want to see how applying \(G\) to the divergences (KL or JS) might behave.

\begin{figure}[ht!]
\begin{center}
\centerline{\includegraphics[width=0.4\textwidth]{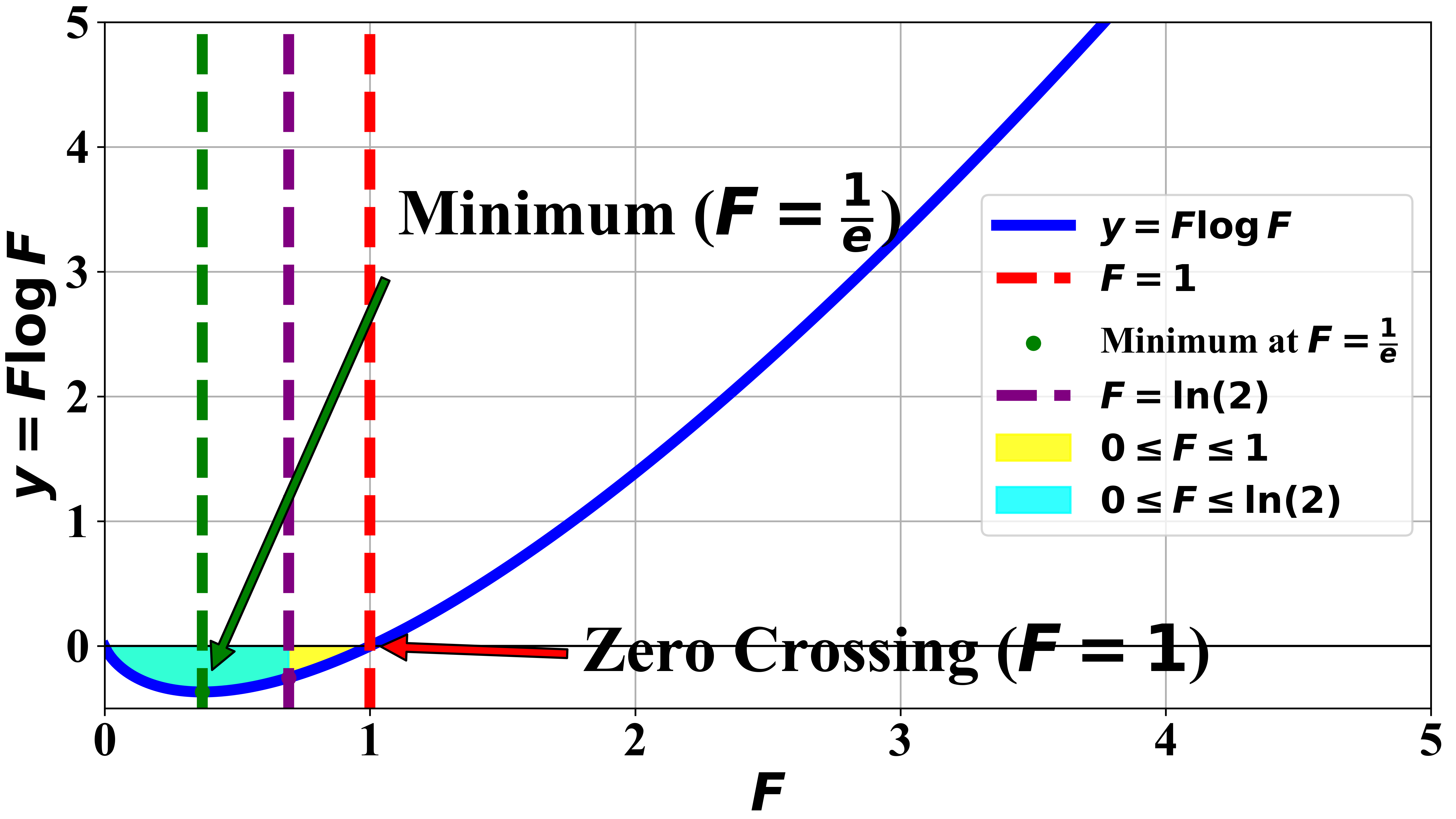}}
\caption{Visual illustration of $F \log F$.}
\label{flogfviso}
\end{center}
\vspace{-20pt}
\end{figure}
\textbf{Properties of \texorpdfstring{\(G(F)\)}{G(F)}} For any positive real number \(F\), the behavior of \(F\,\log F\) depends on whether \(F\) lies above or below unity. In particular, \(F\,\log F < 0\) whenever \(0 < F < 1\), and \(F\,\log F > 0\) for \(F > 1\). At \(F=1\), the function evaluates to zero. Moreover, as \(F \to 0^+\), the term \(F\,\log F\) approaches zero from the negative side, i.e., \(\lim_{F \to 0^+} F \,\log F = 0\). Conversely, \(F\,\log F\) grows unbounded for large \(F\to +\infty\). An illustrative plot of these trends is shown in Fig.~\ref{flogfviso}.

\textbf{Effect on KL Divergence} Because the KL divergence ranges from \(0\) to \(+\infty\), consider 
\begin{equation}
  G\bigl(D_{\mathrm{KL}}(P \| Q)\bigr) \;=\; D_{\mathrm{KL}}(P \| Q)\,\log\Bigl(D_{\mathrm{KL}}(P \| Q)\Bigr).
\end{equation} When \(D_{\mathrm{KL}}(P \| Q)\) is large (e.g., when \(P\) and \(Q\) are significantly different), \(G(\cdot)\) can amplify the difference since \(F \log F\) grows faster than \(F\) for sufficiently large \(F\). When \(D_{\mathrm{KL}}(P \| Q)\) is near zero, the value of \(F \log F\) remains near zero. Hence, applying \(F \log F\) to the KL divergence can provide a sharper penalty when the distributions differ significantly.

\textbf{Effect on JS Divergence} By contrast, the JS divergence satisfies
\[
0 \;\le\; D_{\mathrm{JS}}(P\|Q) \;\le\; \log(2)\approx 0.693,
\]
which is strictly below 1. Thus, for any $F \in (0,\log 2] \subset (0,1)$, we have $\log(F)<0$, implying $F\log F \le 0$. Define
\begin{equation}
  G\bigl(D_{\mathrm{JS}}(P\|Q)\bigr)
  \;=\; 
  D_{\mathrm{JS}}(P\|Q)\,\log\!\bigl(D_{\mathrm{JS}}(P\|Q)\bigr).
\end{equation}
In this bounded interval, $F\log F$ achieves its global minimum around $F=e^{-1}\approx0.3679$, where $F\log F \approx -0.3679$. 
At $F=\log(2)\approx0.693$, the value is about $-0.254$. Hence overall, $F\log F \;\in\; [-0.3679,\,0] \quad \text{for } F\in[0,\log 2].$
This means the magnitude of $F\log F$ remains relatively small compared to the unbounded KL case. Hence, applying $F\log F$ to $D_{\mathrm{JS}}$ does not drastically increase its value, although it can still slightly accentuate differences within the JS range.

\textbf{Range Analysis} For $D_{\mathrm{KL}}\in[0,+\infty),\, G(D_{\mathrm{KL}})$ can become arbitrarily large, while for $D_{\mathrm{JS}}\in[0,\log 2]\subset(0,1)$, $G(D_{\mathrm{JS}})$ stays within $[-0.3679,0]$. Thus, the transformation $F\mapsto F\log F$ significantly amplifies large values in KL but remains subdued for JS.

\begin{figure}[ht!]
\begin{center}
\centerline{\includegraphics[width=0.45\textwidth]{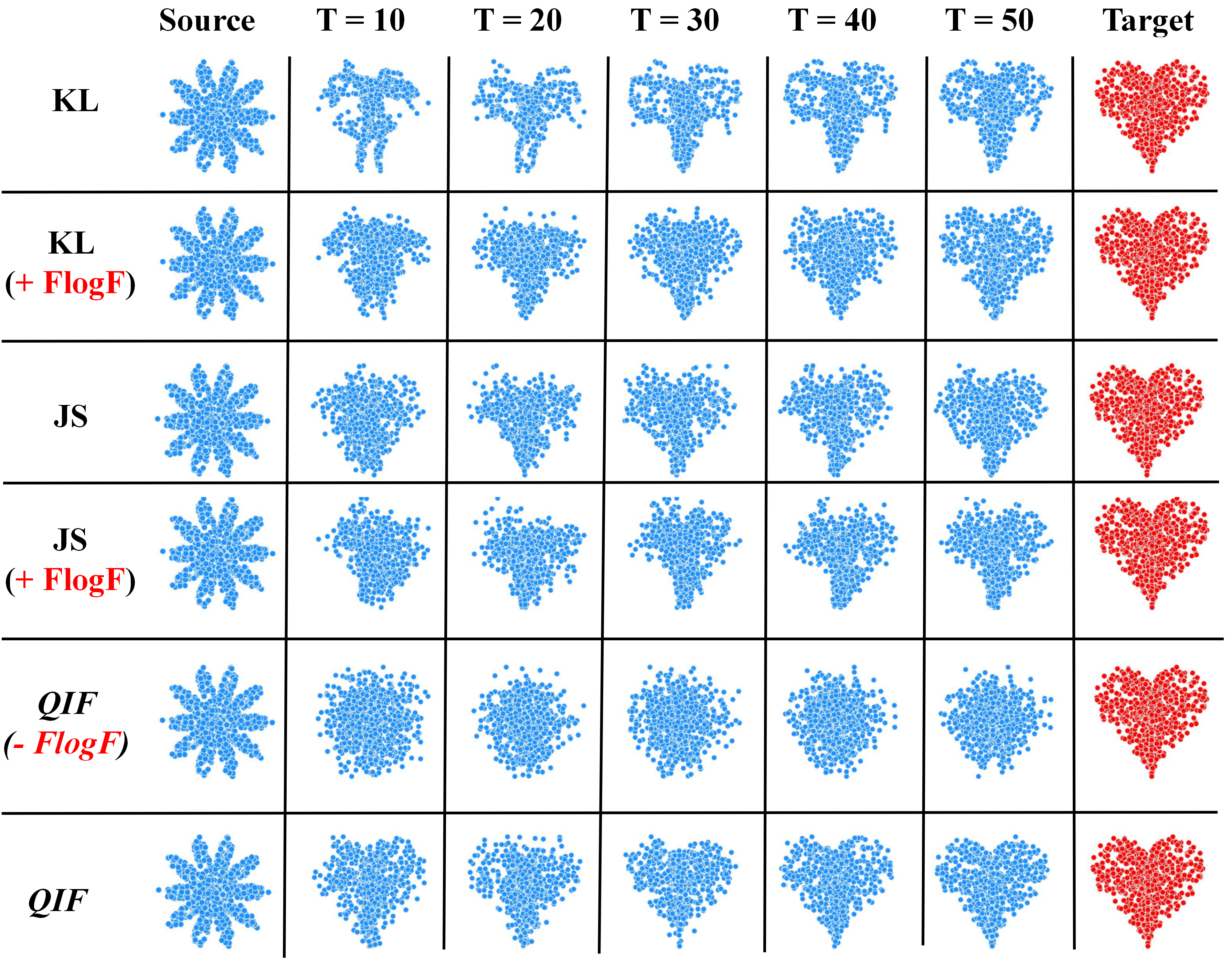}}
\caption{Distribution evolution of different divergence methods (QIF, KL, JS divergence with $F \log F$) and Sinkhorn Distance during optimization. The blue distribution in the beginning stage noted as $\textcolor{blue}{\bm{\ast}}$ is the initial distribution of the input of all algorithms, the red heart distribution noted as $\textcolor{red}{\bm{\heartsuit}}$ is the target distribution, and T is the number of iterations with $\sigma = 0.3$, learning rate $= 0.01$ and $1000$ sample points.}
\label{flogf}
\end{center}
\end{figure}

\textbf{Effect on QIF} We note when $F = 1$, $-F \log(F) = 0$. For intermediate values of $F$, the term $-F \log(F)$ provides a larger gradient, encouraging faster convergence. As $F \to 0^{+}$, we have the limit $\lim_{F \to 0^+} -F \log(F) = 0$, so it remains finite and avoids numerical divergence. We can see why it accelerates training by examining the derivative with respect to $F$. For $\mathcal{L}_{\text{ent}}(F) = -F \log(F)$, we compute $ \frac{d\mathcal{L}_{\text{ent}}}{dF} = -\log(F) - 1$. Compared to $\frac{d}{dF}(1 - F) = -1$, the gradient from $-F \log(F)$ has a \emph{logarithmic} component $-\log(F)$ which grows in magnitude when $F$ is small (yet not pushing the loss to infinity). This $\log$ term drives stronger updates for mid-range or small $F$ values. The function $-x \log(x)$ arises frequently in information theory, for example, in the definition of the Shannon entropy of a probability $p$ \cite{lin1991divergence}: $H(p) = - p \log (p)$. By analogy, we can think of $-F \log(F)$ as placing an \emph{information-theoretic penalty} when the fidelity $F$ differs from 1. Therefore, it fits well with the notion of ``distance'' or ``divergence'' in classical and quantum information contexts.

Including the term $-F \log(F)$ instead of a more straightforward function like $(1 - F)$ often yields better gradient properties, aiding faster and more stable optimization. An entropy-inspired loss that more robustly penalizes states that do not closely match. Hence, empirical results often show that using $-F \log(F)$ improves performance over a simple fidelity-based loss.

\section{The Role of $F \log F$ in Gradients: $\log F + 1$}
To understand the impact of the \( F \log F \) transformation on gradient descent, we derive the gradients of the transformed divergences with respect to model parameters \( \theta \). We define a transformation \( G(F) = F \log F \), where \( F \) represents a divergence measure (either KL or JS divergence, QIF). Applying this transformation aims to modify the behavior of the divergence, potentially enhancing optimization characteristics such as convergence speed and stability.

Assume \( P \) and \( Q \) are parameterized distributions \( P_\theta \) and \( Q_\theta \), where \( \theta \) represents the model parameters. Applying the \( F \log F \) transformation, the gradient formula of JS becomes: 
\begin{equation}
\nabla_\theta G(D_{\mathrm{JS}}(P_\theta \| Q_\theta)) = \left( \log(D_{\mathrm{JS}}) + 1 \right) \cdot \nabla_\theta D_{\mathrm{JS}}. 
\end{equation} Similarly, for KL divergence: \begin{equation}
\nabla_\theta G(D_{\mathrm{KL}}(P_\theta \| Q_\theta)) = \left( \log(D_{\mathrm{KL}}) + 1 \right) \cdot \nabla_\theta D_{\mathrm{KL}}.
\end{equation}

KL divergence has an unbounded range \( [0, +\infty) \). When applying the \( F \log F \) transformation, the scaling factor \( \log(F) + 1 \) introduces significant changes based on the value of \( F \). For large \( D_{\mathrm{KL}} \),
$\log(D_{\mathrm{KL}}) + 1$ increases significantly, which results in an amplified gradient, leading to larger parameter updates and faster convergence when the model is far from the target distribution. For small \( D_{\mathrm{KL}} \), $\log(D_{\mathrm{KL}}) + 1 $ decreases or becomes negative, which reduces the gradient magnitude or even reverses its direction, enhancing stability and preventing overshooting as the model approaches the optimal parameters. Thus, the \( F \log F \) transformation dynamically adjusts the gradient based on the current value of KL divergence, facilitating both rapid convergence and stable refinement.

In contrast, the the JS divergence is bounded within the range \( [0, \log 2] \). Given that \( D_{\mathrm{JS}} \in [0, \log 2] \), the scaling factor \( \log(D_{\mathrm{JS}}) + 1 \) behaves as:  for maximum at \( D_{\mathrm{JS}} = \log 2 \): $\log(\log 2) + 1 \approx 0.634$. 

As \( D_{\mathrm{JS}} \) approaches 0: $\log(D_{\mathrm{JS}}) + 1$ $\text{approaches } -\infty$. 
However, in practice, a small epsilon is added to prevent numerical instability $\log(D_{\mathrm{JS}} + \epsilon) + 1$. Thus, the scaling factor remains within a limited range.
\begin{figure}[ht!]
\begin{center}
\centerline{\includegraphics[width=0.4\textwidth]{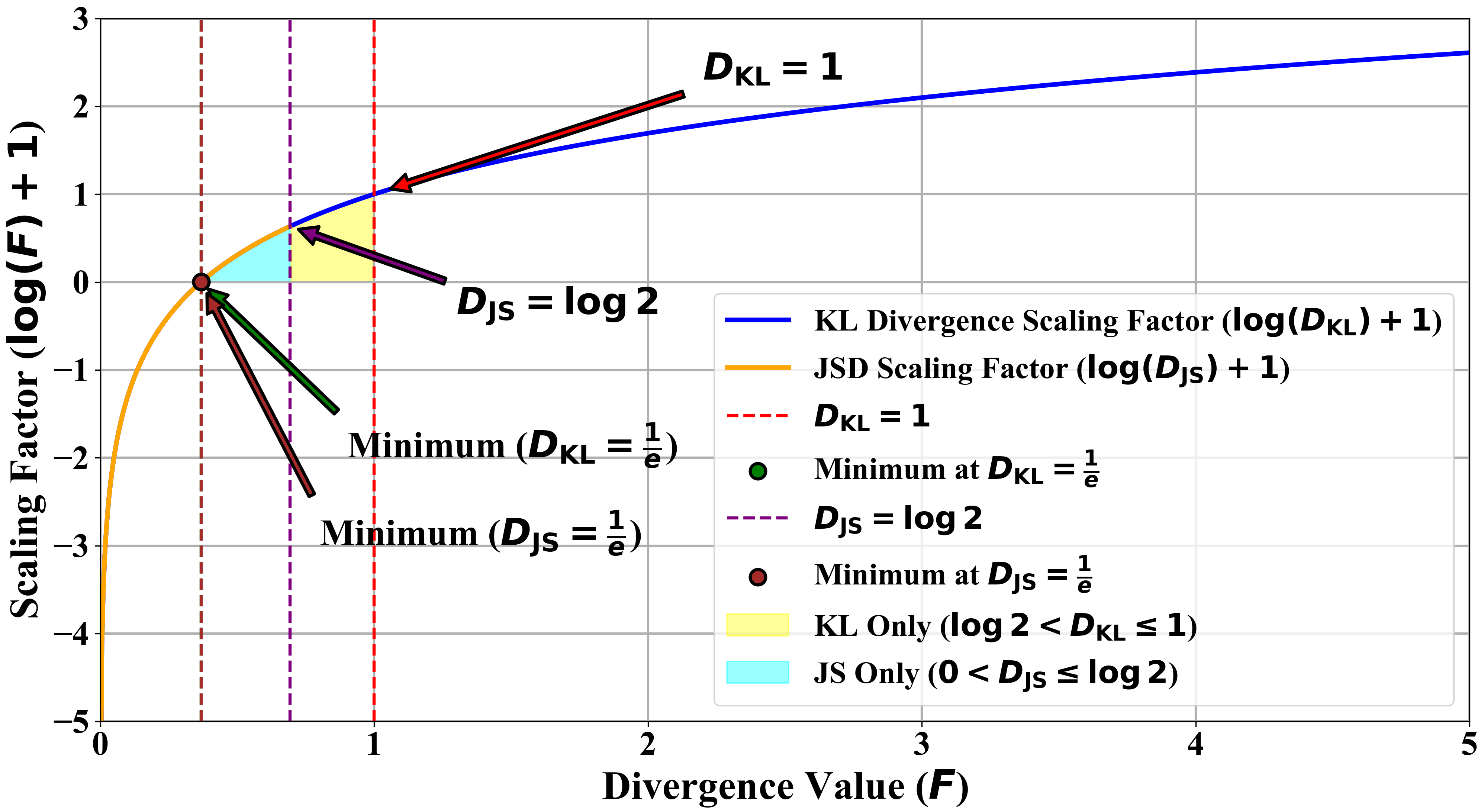}}
\caption{Visual illustration of $ \log (F) + 1 $.}
\label{flogfvisograndient}
\end{center}
\vspace{-25pt}
\end{figure}

Due to the limited range of \( D_{\mathrm{JS}} \), the scaling factor \( \log(D_{\mathrm{JS}}) + 1 \) does not vary as dramatically as it does for KL divergence. Consequently, the \( F \log F \) transformation induces only minor adjustments to the gradient of JS divergence, resulting in a negligible impact on the optimization process, as shown in Fig. \ref{flogfvisograndient}. 

So, for KL Divergence, the unbounded nature allows for significant dynamic adjustment of gradients, facilitating both rapid convergence and stable optimization. However, for JS divergence, the bounded range restricts the extent of gradient scaling, resulting in minimal changes to the optimization dynamics. As shown in Fig. \ref{flogf}, the $F \log F$ brings better convergence to KL, but for JS, this benefit is weak. At the same time, we conducted a comparative experiment for the proposed QIF. However, QIF without $F \log F$ shows poor convergence properties. Therefore, for different algorithms, due to the characteristics of their boundary, $F \log F$ will also show different benefits. 
% Figure. \ref{flogfvisograndient} shows a significant increase in the scaling factor, leading to amplified gradients and accelerated convergence. Small \( F \) reduction or reversal of the scaling factor, enhancing stability near optimal parameters. However, for JS divergence, whose maximum \( F \): scaling factor reaches approximately a moderate value $0.634$. And minimum \( F \): With epsilon adjustment, the scaling factor does not reach extremely negative values, thus limiting its impact.

\section{Application to Deep Neural Network}
To further verify the effectiveness of the proposed QIF in the deep learning framework, we tested the performance of QR-Drop based on QIF on different public datasets.

\subsection{Classical R-Drop Regularization}
Given the training dataset $D = \{(x_i, y_i)\}_{i=1}^{n}$, the goal of training is to learn a model $P_w(y|x)$, where $n$ is the number of training samples, and $(x_i, y_i)$ is a labeled data pair. The main learning objective is to minimize the negative log-likelihood (NLL) loss function
$L_{\text{nll}}$. R-Drop forces two distributions for the same data sample outputted by the two sub-models to be consistent with each other by minimizing the bidirectional KL divergence between the two distributions:
\begin{align}
L_{\text{KL}} = \frac{1}{2} \Big( & D_{\text{KL}}(P_{w_1}(y_i|x_i) \parallel P_{w_2}(y_i|x_i)) \notag \\
& + D_{\text{KL}}(P_{w_2}(y_i|x_i) \parallel P_{w_1}(y_i|x_i)) \Big),
\end{align}
where $P_{w_1}(y_i|x_i)$ and $P_{w_2}(y_i|x_i)$ are the output distributions of two different sub-models produced by dropout for the same input $x_i$. The final training objective is to minimize the combination of the NLL loss and the KL-divergence loss for each data sample $(x_i, y_i)$:
\begin{equation}
\label{eq:lossrdrop}
L_i  = L_{\text{NLL}} + \beta \cdot L_{\text{KL}},
\end{equation} where $\beta$ is a coefficient weight to control the contribution of the KL-divergence loss.

\subsection{QIF for R-drop (QR-Drop) Regularization}
\begin{figure}[ht]
\begin{center}
\centerline{\includegraphics[width=0.4\textwidth]{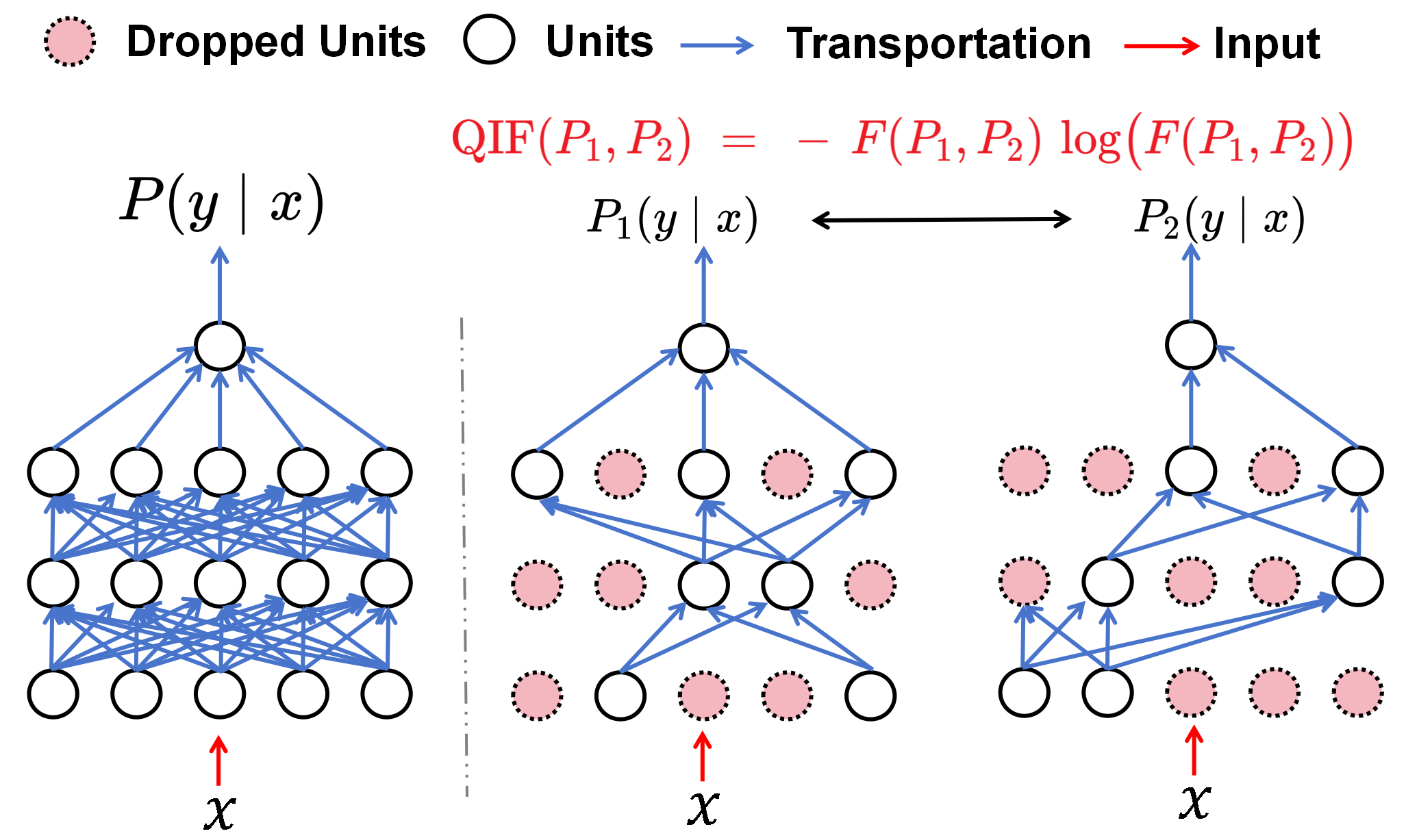}}
\caption{The overall framework of QR-Drop. \textbf{Left}: Neural network model. \textbf{Right}: QIF based QR-Drop. The difference in fidelity $\text{QIF}(P_1,P_2)$ is calculated for $L_{QIF}$.}
\label{Mainflow}
\end{center}
\vspace{-20pt}
\end{figure}
We propose a novel dropout method based on QIF. By integrating quantum-inspired fidelity into the R-Drop framework, we can impose a stricter consistency constraint. Following Eq. \ref{eq:lossrdrop}, the total loss becomes:
\begin{equation}
L_{i}  = L_{\text{NLL}} + \beta \cdot L_{\text{QIF}}, 
\end{equation} where \( \beta\geq 0 \) is a hyperparameter controlling the weight of the fidelity term, as shown in Fig. \ref{Mainflow}.  

\begin{figure}[h]
    \centering
    \begin{subfigure}[b]{0.23\textwidth}
        \centering
        \includegraphics[width=\textwidth]{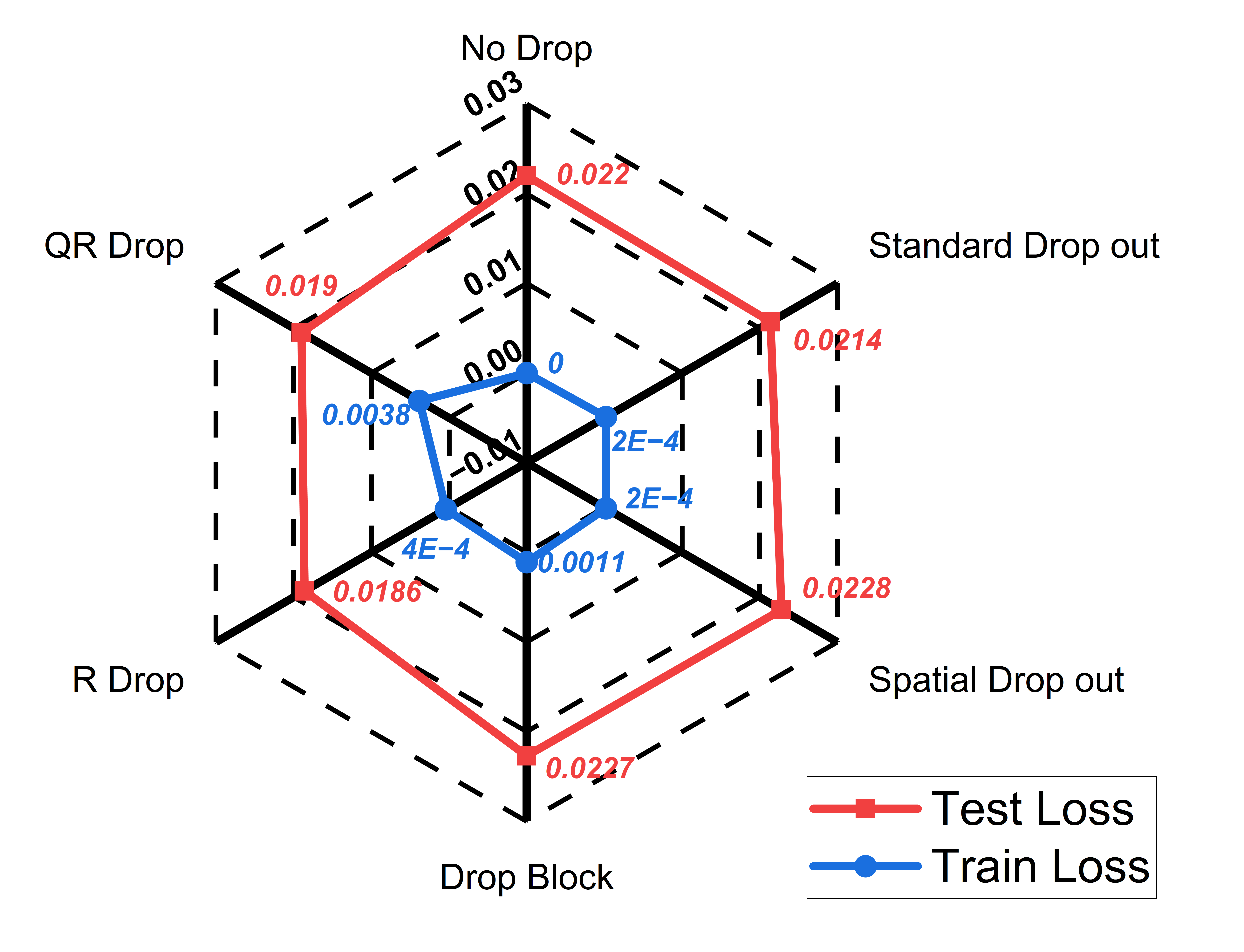}
        \caption{Loss radar curve}
    \end{subfigure}
    \begin{subfigure}[b]{0.23\textwidth}
        \centering
        \includegraphics[width=\textwidth]{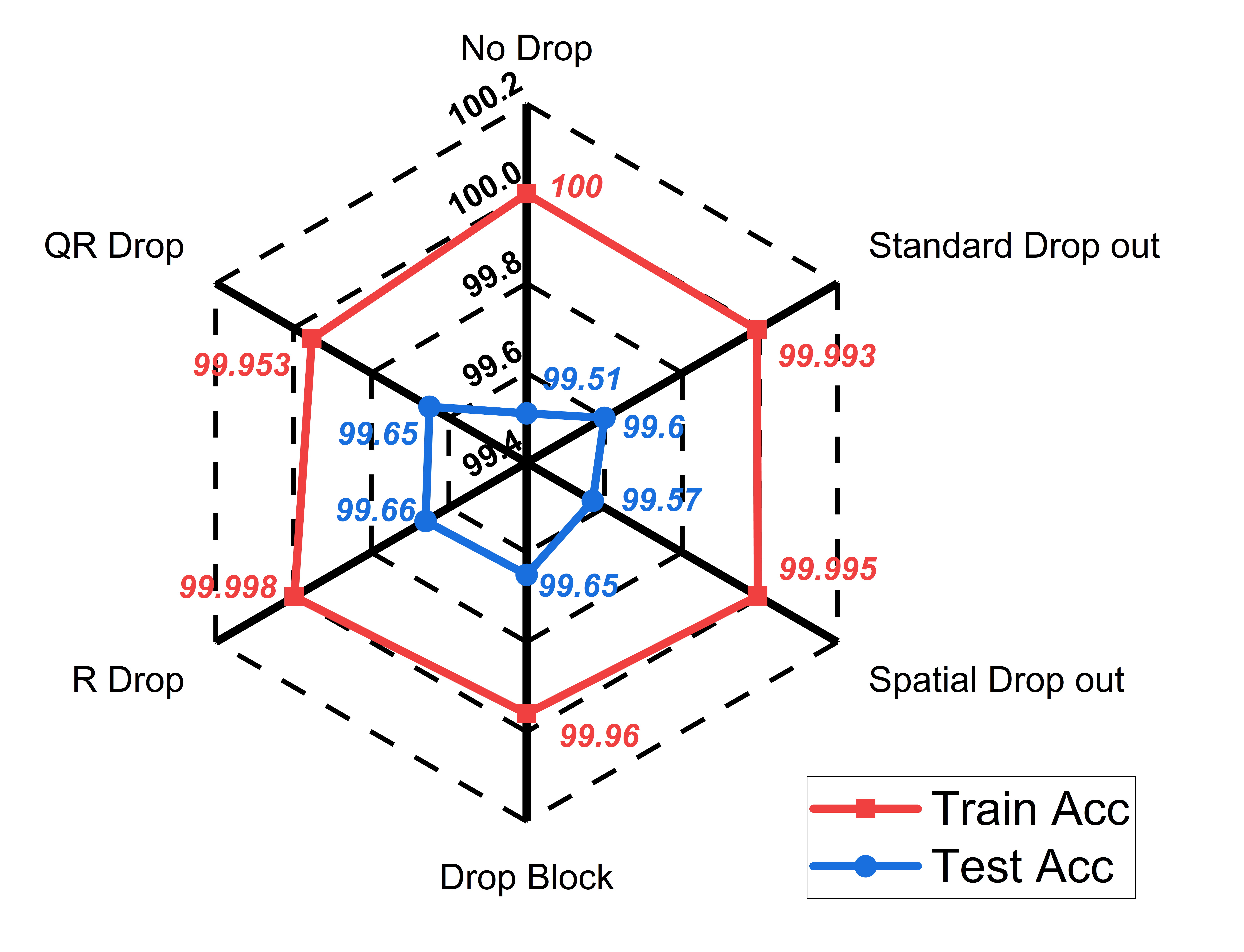}
        \caption{Accuracy radar curve}
    \end{subfigure}
    \caption{Comparison with dropout methods on MNIST.}
    \label{fig:qr-dropmnist}
    \vspace{-15pt}
\end{figure}

\subsection{Image Classification Task}
\textbf{Settings} We adopt the stochastic gradient descent (SGD) \cite{robbins1951stochastic} optimizer. The initial learning rate $\eta$ is set to $1e^{-1}$ and decreased by $10$ every $60$ epoch for the CIFAR-10 \cite{krizhevsky2009learning} data set and MNIST \cite{mnist} with the backbone ResNet-18 \cite{he2016deepresnet}. The dropout rate is fixed at $0.1$.
\begin{figure}[h]
    \centering
    \begin{subfigure}[b]{0.23\textwidth}
        \centering
        \includegraphics[width=\textwidth]{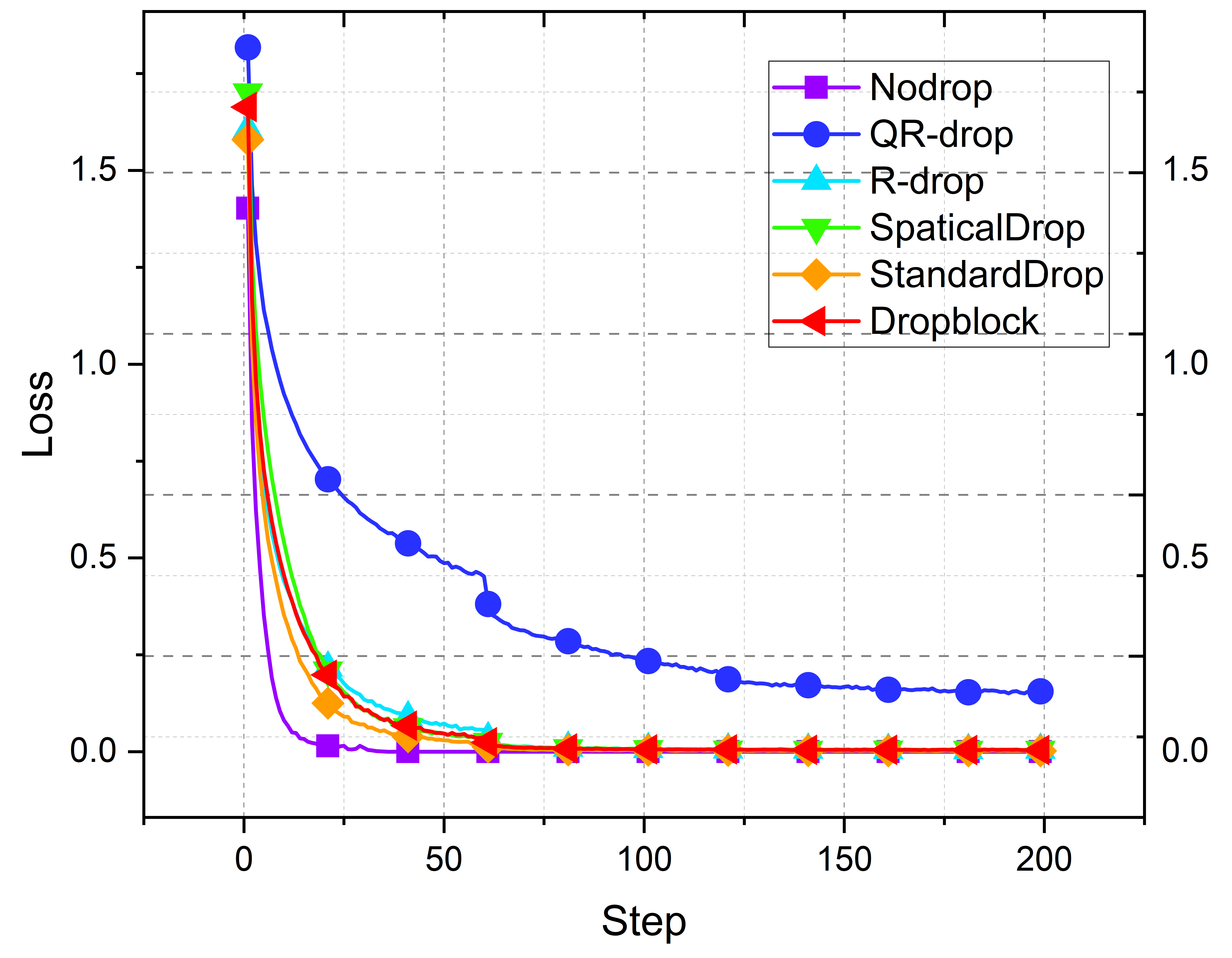}
        \caption{Train loss curve}
    \end{subfigure}
    \begin{subfigure}[b]{0.24\textwidth}
        \centering
        \includegraphics[width=\textwidth]{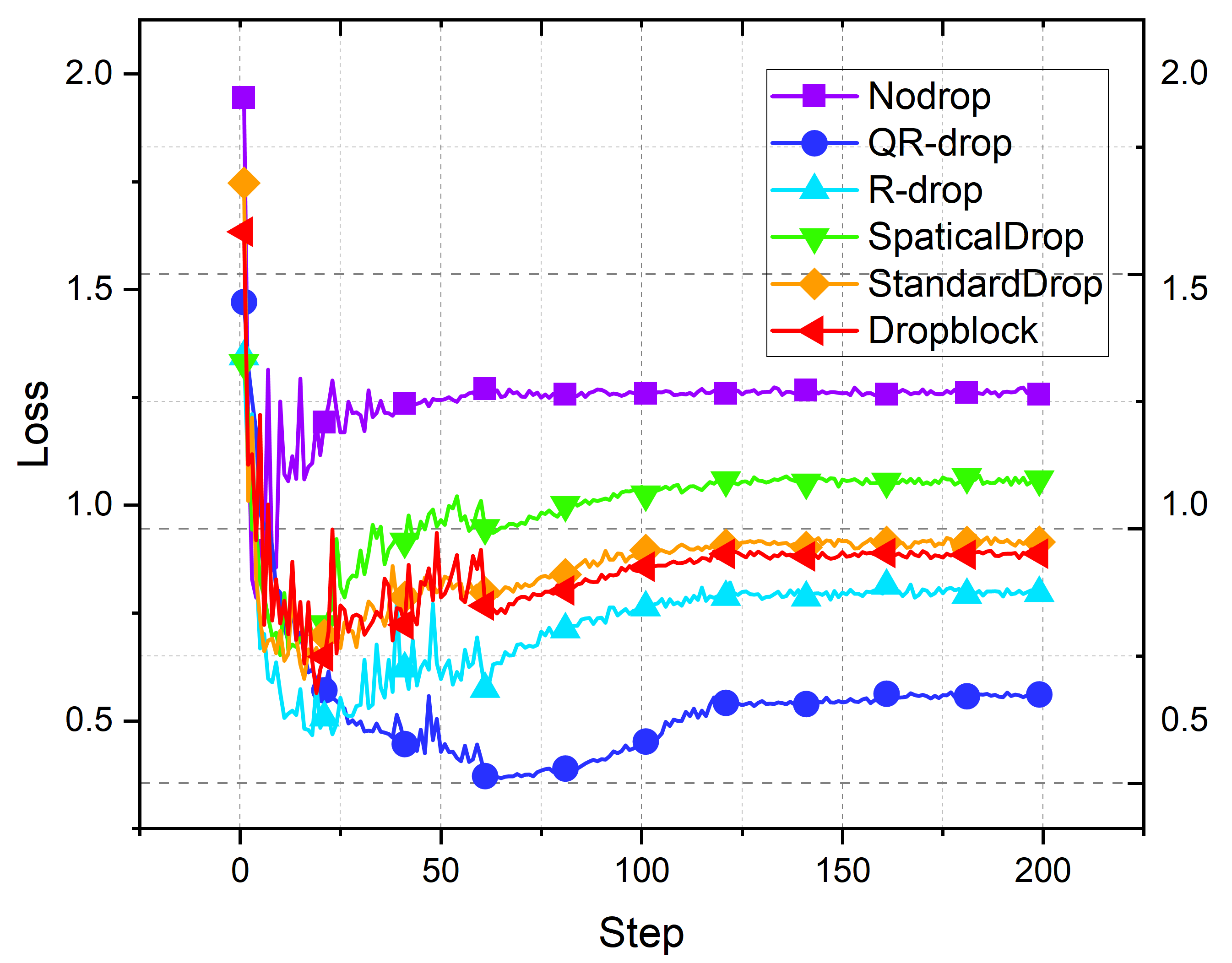}
        \caption{Test loss curve}
    \end{subfigure}
        \caption{Comparison with dropout methods on CIFAR-10.}
    \label{fig:qr-dropcifar10loss}
\end{figure}

The results in Fig. \ref{fig:qr-dropmnist}-\ref{fig:qr-cifar10dropacc}  show that during training, the unregularized ResNet-18 is more susceptible to overfitting, as evidenced by its convergence to the highest test loss despite achieving the lowest training loss. In contrast, R-Drop yields better convergence than the standard drop method, but QR-Drop demonstrates the best overall performance. 

\begin{figure}[h]
    \centering
    \begin{subfigure}[b]{0.23\textwidth}
        \centering
        \includegraphics[width=\textwidth]{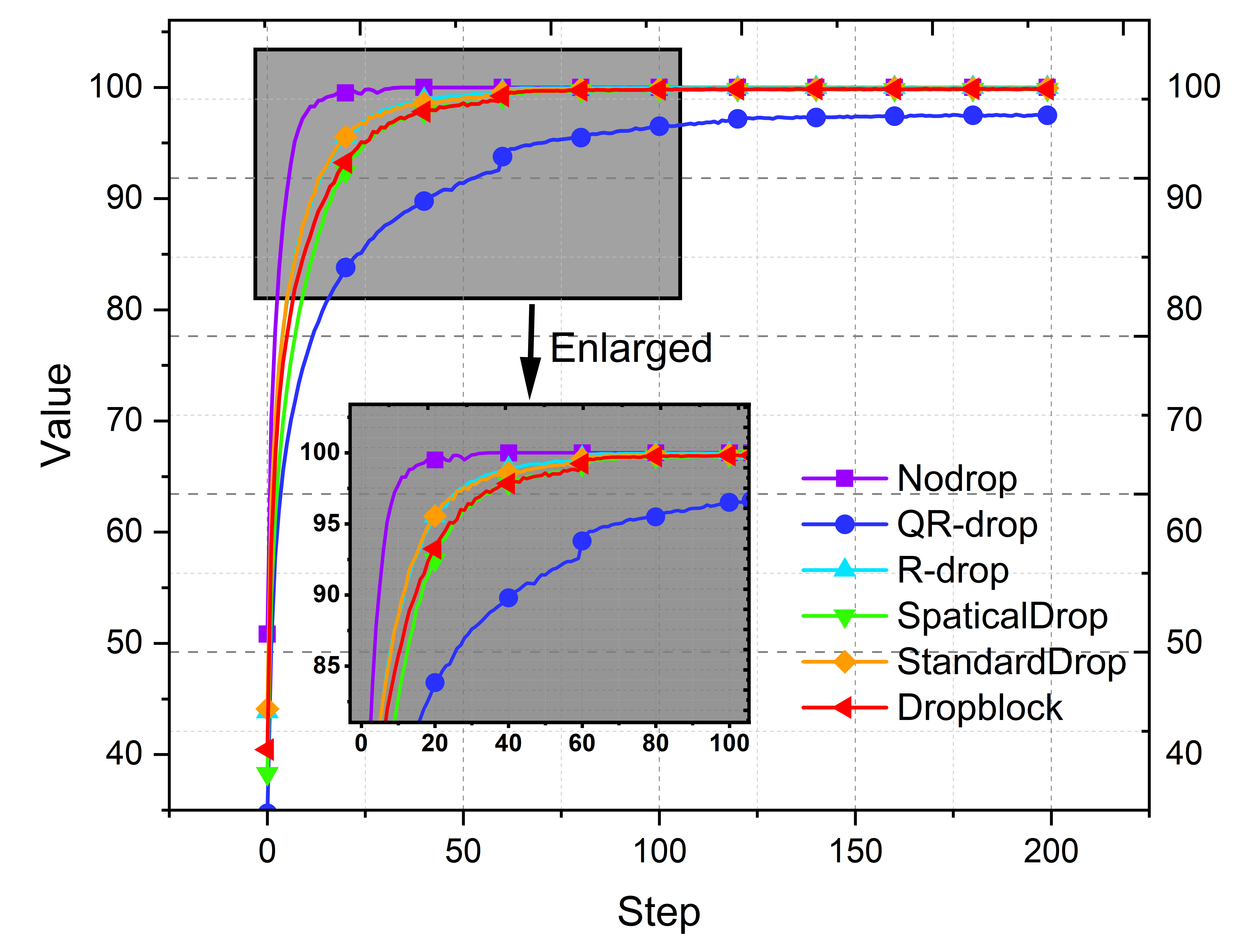}
        \caption{Train accuracy curve}
    \end{subfigure}
    \begin{subfigure}[b]{0.23\textwidth}
        \centering
        \includegraphics[width=\textwidth]{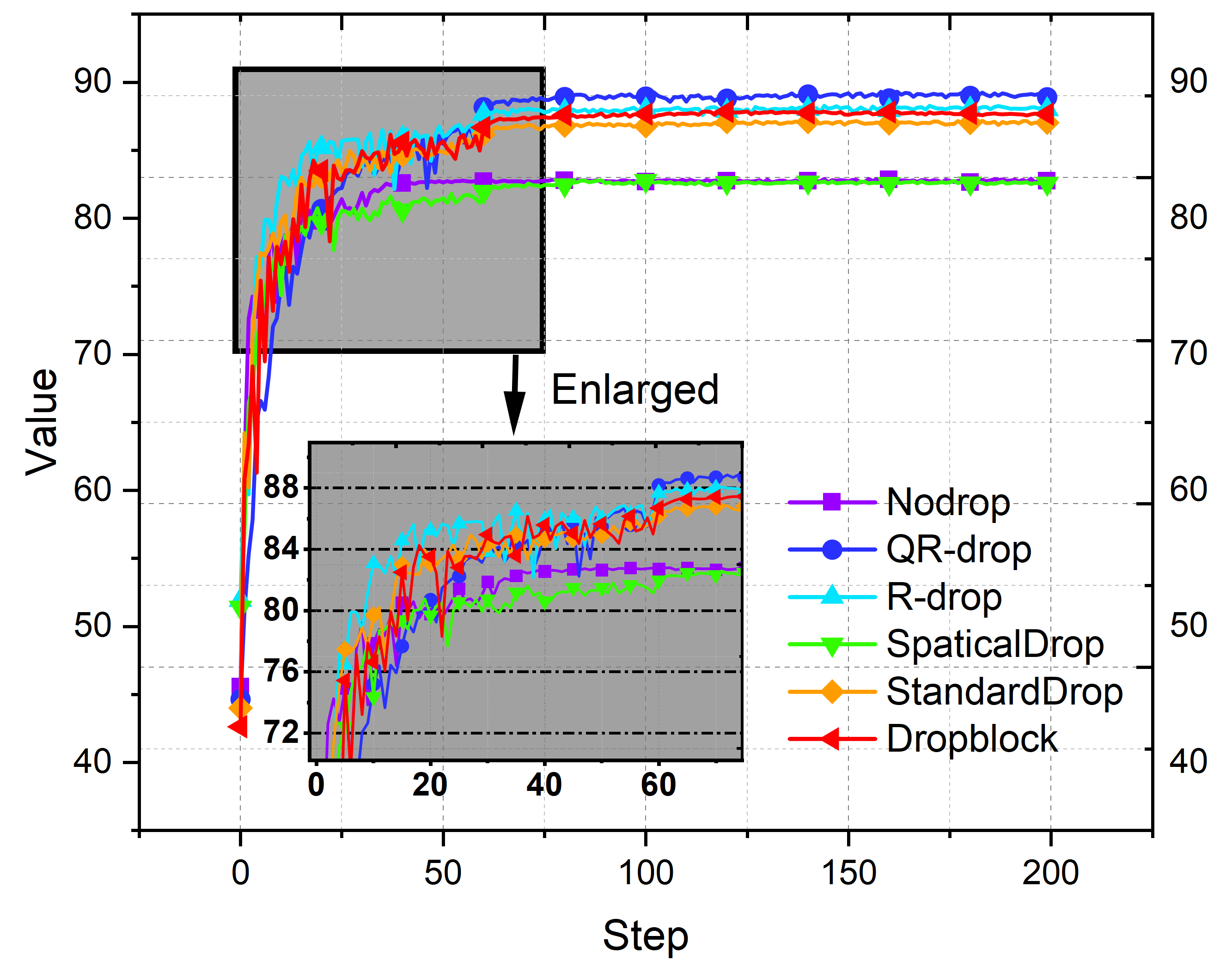}
        \caption{Test accuracy curve}
    \end{subfigure}
    \caption{Comparison with dropout methods on CIFAR-10.}
    \label{fig:qr-cifar10dropacc}
\end{figure}

\subsection{Language Understanding Task}
We take the BERT-base \cite{devlin2018bert}, ELECTRA-large \cite{clark2020electra}, and RoBERTa-large \cite{liu2019roberta} pre-trained models as our backbones to perform fine-tuning. The dataset is the commonly adopted GLUE benchmark \cite{wang2018glue}. The result for STS-B is the Pearson correlation; Matthew’s correlation is used for CoLA; Other tasks are measured by Accuracy. QR-Drop yields consistent gains or matches performance across a variety of GLUE tasks compared to both the baseline and the R-Drop method, as observed by higher average scores for BERT-base (78.5 vs.\ 78.2) and ELECTRA-large (87.0 vs.\ 86.9). In particular, it improves upon challenging tasks like RTE (BERT-base) and CoLA (ELECTRA-large), highlighting its effectiveness in enhancing model generalization in Table. \ref{table:glue-results}.

\begin{table}[ht]
    \centering
    \setlength{\tabcolsep}{2pt}
    \fontsize{8pt}{4pt}\selectfont
    \begin{tabular}{l|c|c|c|c|c|c}
        \toprule
        \textbf{Model} & \textbf{MNLI} & \textbf{RTE} & \textbf{QQP} & \textbf{SST-2} & \textbf{CoLA} & \textbf{Avg} \\
        \midrule
        BERT-base & $83.6$ & $61.4 $ & $90.1 $ & $92.3$ & $59.1 $ & $77.2 $ \\
        \addlinespace[2pt]
        BERT-base + RD & $84.4$ & $61.4 $ & $91.1 $ & $92.8$ & $60.4$ & $78.2$ \\
        \addlinespace[2pt]
        \hline
        \addlinespace[2pt]
        \textbf{BERT-base + QRD} & $84.6$ & $63.2 $ & $91.1 $ & $92.8$ & $60.7 $ & $\textbf{78.5} $\\
        \midrule
        ELECTRA-large & $90.6$ & $85.0 $ & $92.2 $ & $96.4$ & $69.0 $ & $86.6$ \\
        \addlinespace[2pt]
        ELECTRA-large + RD & $90.6$ & $85.6 $ & $92.7 $ & $96.8$ & $69.1$ & $86.9$ \\
        \addlinespace[2pt]
        \hline
        \addlinespace[2pt]
        \textbf{ELECTRA-large + QRD} & $90.6$ & $85.4 $ & $92.5 $ & $96.7$ & $69.4 $ & $\textbf{87.0} $ \\
        \midrule
        RoBERTa-large & $90.1$ & $85.2 $ & $92.2 $ & $95.6$ & $66.4 $ & $85.9 $ \\
        \addlinespace[2pt]
        RoBERTa-large + RD & $90.6$ & $86.7 $ & $92.3 $ & $96.1$ & $67.1$ & $86.6^\ast $ \\
        \addlinespace[2pt]
        \hline
        \addlinespace[2pt]
         \textbf{RoBERTa-large + QRD} & $90.7$ & $87.1 $ & $92.5 $ & $96.4$ & $66.9 $ & $86.6^\ddagger $\\
        \bottomrule
    \end{tabular}
    \caption{Performance of fine-tuned models on GLUE. RD is R-Drop, and QRD is QR-Drop with $\beta = 1.0$. $\ast$:86.56 and $\ddagger$:86.62.}
    \label{table:glue-results}
  
\end{table}

\textbf{Limitations}  Current experiments are limited to classification tasks; QIF's effectiveness in reinforcement learning, generative models, and self-supervised learning remains unexplored. It also lacks evaluations on large-scale NLP and vision models (e.g., GPT, ViT). In future work, we plan to expand QIF applications, assess its impact on broader tasks, and explore its potential benefits in deep learning.

\section{Conclusion}
In this work, we present the quantum-inspired fidelity-based divergence (QIF), a novel algorithm that overcomes the key limitations of the conventional Kullback–Leibler divergence—particularly its instability when distributions differ significantly or contain zero probabilities—while maintaining the same \(\mathcal{O}(n)\) computational complexity. Compared to the standard quantum relative entropy, QIF is a more adaptive and stable alternative for machine learning tasks. Furthermore, we introduce QR-Drop, a new dropout algorithm derived from QIF, demonstrating superior robustness over the traditional R-Drop. Notably, the proposed quantum fidelity-inspired approach does not require validation on actual quantum hardware, thereby circumventing the constraints posed by current quantum computing capabilities. This hardware-independence reinforces QIF’s practical and theoretical advantages across a wide range of applications.
\clearpage

\bibliography{Reference}

\begin{thebibliography}{57}
\providecommand{\natexlab}[1]{#1}
\providecommand{\url}[1]{\texttt{#1}}
\expandafter\ifx\csname urlstyle\endcsname\relax
  \providecommand{\doi}[1]{doi: #1}\else
  \providecommand{\doi}{doi: \begingroup \urlstyle{rm}\Url}\fi

\bibitem[Arbel et~al.(2019)Arbel, Korba, Salim, and Gretton]{arbel2019maximummmd}
Arbel, M., Korba, A., Salim, A., and Gretton, A.
\newblock Maximum mean discrepancy gradient flow.
\newblock \emph{Advances in Neural Information Processing Systems}, 32, 2019.

\bibitem[Ba \& Frey(2013)Ba and Frey]{ba2013adaptive}
Ba, J. and Frey, B.
\newblock Adaptive dropout for training deep neural networks.
\newblock \emph{Advances in neural information processing systems}, 26, 2013.

\bibitem[Bach(2022)]{kkl}
Bach, F.
\newblock Information theory with kernel methods.
\newblock \emph{IEEE Transactions on Information Theory}, 69\penalty0 (2):\penalty0 752--775, 2022.

\bibitem[Baldi \& Sadowski(2013)Baldi and Sadowski]{baldi2013understanding}
Baldi, P. and Sadowski, P.~J.
\newblock Understanding dropout.
\newblock \emph{Advances in neural information processing systems}, 26, 2013.

\bibitem[Bengtsson \& {\.Z}yczkowski(2017)Bengtsson and {\.Z}yczkowski]{bengtsson2017geometry}
Bengtsson, I. and {\.Z}yczkowski, K.
\newblock \emph{Geometry of quantum states: an introduction to quantum entanglement}.
\newblock Cambridge university press, 2017.

\bibitem[Bergholm et~al.(2022)Bergholm, Izaac, Schuld, Gogolin, Ahmed, Ajith, Alam, Alonso-Linaje, AkashNarayanan, Asadi, Arrazola, Azad, Banning, Blank, Bromley, Cordier, Ceroni, Delgado, Matteo, Dusko, Garg, Guala, Hayes, Hill, Ijaz, Isacsson, Ittah, Jahangiri, Jain, Jiang, Khandelwal, Kottmann, Lang, Lee, Loke, Lowe, McKiernan, Meyer, Montañez-Barrera, Moyard, Niu, O'Riordan, Oud, Panigrahi, Park, Polatajko, Quesada, Roberts, Sá, Schoch, Shi, Shu, Sim, Singh, Strandberg, Soni, Száva, Thabet, Vargas-Hernández, Vincent, Vitucci, Weber, Wierichs, Wiersema, Willmann, Wong, Zhang, and Killoran]{pennylane}
Bergholm, V., Izaac, J., Schuld, M., Gogolin, C., Ahmed, S., Ajith, V., Alam, M.~S., Alonso-Linaje, G., AkashNarayanan, B., Asadi, A., Arrazola, J.~M., Azad, U., Banning, S., Blank, C., Bromley, T.~R., Cordier, B.~A., Ceroni, J., Delgado, A., Matteo, O.~D., Dusko, A., Garg, T., Guala, D., Hayes, A., Hill, R., Ijaz, A., Isacsson, T., Ittah, D., Jahangiri, S., Jain, P., Jiang, E., Khandelwal, A., Kottmann, K., Lang, R.~A., Lee, C., Loke, T., Lowe, A., McKiernan, K., Meyer, J.~J., Montañez-Barrera, J.~A., Moyard, R., Niu, Z., O'Riordan, L.~J., Oud, S., Panigrahi, A., Park, C.-Y., Polatajko, D., Quesada, N., Roberts, C., Sá, N., Schoch, I., Shi, B., Shu, S., Sim, S., Singh, A., Strandberg, I., Soni, J., Száva, A., Thabet, S., Vargas-Hernández, R.~A., Vincent, T., Vitucci, N., Weber, M., Wierichs, D., Wiersema, R., Willmann, M., Wong, V., Zhang, S., and Killoran, N.
\newblock Pennylane: Automatic differentiation of hybrid quantum-classical computations, 2022.
\newblock URL \url{https://arxiv.org/abs/1811.04968}.

\bibitem[Bhattacharyya(1943)]{bhattacharyya1943measure}
Bhattacharyya, A.
\newblock On a measure of divergence between two statistical populations defined by their probability distribution.
\newblock \emph{Bulletin of the Calcutta Mathematical Society}, 35:\penalty0 99--110, 1943.

\bibitem[Blei et~al.(2017)Blei, Kucukelbir, and McAuliffe]{blei2017variational}
Blei, D.~M., Kucukelbir, A., and McAuliffe, J.~D.
\newblock Variational inference: A review for statisticians.
\newblock \emph{Journal of the American statistical Association}, 112\penalty0 (518):\penalty0 859--877, 2017.

\bibitem[Bonnici(2020)]{bonnici2020kullback}
Bonnici, V.
\newblock Kullback-leibler divergence between quantum distributions, and its upper-bound.
\newblock \emph{arXiv preprint arXiv:2008.05932}, 2020.

\bibitem[Chazal et~al.(2024)Chazal, Korba, and Bach]{regularizedkkl}
Chazal, C., Korba, A., and Bach, F.
\newblock Statistical and geometrical properties of regularized kernel kullback-leibler divergence.
\newblock \emph{arXiv preprint arXiv:2408.16543}, 2024.

\bibitem[Choe \& Shim(2019)Choe and Shim]{choe2019attention}
Choe, J. and Shim, H.
\newblock Attention-based dropout layer for weakly supervised object localization.
\newblock In \emph{Proceedings of the IEEE/CVF conference on computer vision and pattern recognition}, pp.\  2219--2228, 2019.

\bibitem[Ciliberto et~al.(2018)Ciliberto, Herbster, Ialongo, Pontil, Rocchetto, Severini, and Wossnig]{ciliberto2018quantum}
Ciliberto, C., Herbster, M., Ialongo, A.~D., Pontil, M., Rocchetto, A., Severini, S., and Wossnig, L.
\newblock Quantum machine learning: a classical perspective.
\newblock \emph{Proceedings of the Royal Society A: Mathematical, Physical and Engineering Sciences}, 474\penalty0 (2209):\penalty0 20170551, 2018.

\bibitem[Clark(2020)]{clark2020electra}
Clark, K.
\newblock Electra: Pre-training text encoders as discriminators rather than generators.
\newblock \emph{arXiv preprint arXiv:2003.10555}, 2020.

\bibitem[Dauphin et~al.(2014)Dauphin, Pascanu, Gulcehre, Cho, Ganguli, and Bengio]{dauphin2014identifying}
Dauphin, Y.~N., Pascanu, R., Gulcehre, C., Cho, K., Ganguli, S., and Bengio, Y.
\newblock Identifying and attacking the saddle point problem in high-dimensional non-convex optimization.
\newblock \emph{Advances in neural information processing systems}, 27, 2014.

\bibitem[Deng(2012)]{mnist}
Deng, L.
\newblock The mnist database of handwritten digit images for machine learning research [best of the web].
\newblock \emph{IEEE Signal Processing Magazine}, 29\penalty0 (6):\penalty0 141--142, 2012.
\newblock \doi{10.1109/MSP.2012.2211477}.

\bibitem[Devlin(2018)]{devlin2018bert}
Devlin, J.
\newblock Bert: Pre-training of deep bidirectional transformers for language understanding.
\newblock \emph{arXiv preprint arXiv:1810.04805}, 2018.

\bibitem[Ding et~al.(2020)Ding, Wu, Holmes, Wiseth, Franklin, Martonosi, and Chong]{ding2020square}
Ding, Y., Wu, X.-C., Holmes, A., Wiseth, A., Franklin, D., Martonosi, M., and Chong, F.~T.
\newblock Square: Strategic quantum ancilla reuse for modular quantum programs via cost-effective uncomputation.
\newblock In \emph{2020 ACM/IEEE 47th Annual International Symposium on Computer Architecture (ISCA)}, pp.\  570--583. IEEE, 2020.

\bibitem[Gal \& Ghahramani(2016)Gal and Ghahramani]{gal2016theoretically}
Gal, Y. and Ghahramani, Z.
\newblock A theoretically grounded application of dropout in recurrent neural networks.
\newblock \emph{Advances in neural information processing systems}, 29, 2016.

\bibitem[Gal et~al.(2017)Gal, Hron, and Kendall]{gal2017concrete}
Gal, Y., Hron, J., and Kendall, A.
\newblock Concrete dropout.
\newblock \emph{Advances in neural information processing systems}, 30, 2017.

\bibitem[Gelman et~al.(1995)Gelman, Carlin, Stern, and Rubin]{gelman1995bayesian}
Gelman, A., Carlin, J.~B., Stern, H.~S., and Rubin, D.~B.
\newblock \emph{Bayesian data analysis}.
\newblock Chapman and Hall/CRC, 1995.

\bibitem[Ghiasi et~al.(2018)Ghiasi, Lin, and Le]{ghiasi2018dropblock}
Ghiasi, G., Lin, T.-Y., and Le, Q.~V.
\newblock Dropblock: A regularization method for convolutional networks.
\newblock \emph{Advances in neural information processing systems}, 31, 2018.

\bibitem[Glaser et~al.(2021)Glaser, Arbel, and Gretton]{glaser2021kale}
Glaser, P., Arbel, M., and Gretton, A.
\newblock Kale flow: A relaxed kl gradient flow for probabilities with disjoint support.
\newblock \emph{Advances in Neural Information Processing Systems}, 34:\penalty0 8018--8031, 2021.

\bibitem[Gretton et~al.(2012)Gretton, Borgwardt, Rasch, Sch{\"o}lkopf, and Smola]{gretton2012kernel}
Gretton, A., Borgwardt, K.~M., Rasch, M.~J., Sch{\"o}lkopf, B., and Smola, A.
\newblock A kernel two-sample test.
\newblock \emph{The Journal of Machine Learning Research}, 13\penalty0 (1):\penalty0 723--773, 2012.

\bibitem[He et~al.(2016)He, Zhang, Ren, and Sun]{he2016deepresnet}
He, K., Zhang, X., Ren, S., and Sun, J.
\newblock Deep residual learning for image recognition.
\newblock In \emph{Proceedings of the IEEE conference on computer vision and pattern recognition}, pp.\  770--778, 2016.

\bibitem[Horodecki et~al.(2009)Horodecki, Horodecki, Horodecki, and Horodecki]{horodecki2009quantum}
Horodecki, R., Horodecki, P., Horodecki, M., and Horodecki, K.
\newblock Quantum entanglement.
\newblock \emph{Reviews of modern physics}, 81\penalty0 (2):\penalty0 865--942, 2009.

\bibitem[Kingma(2013)]{kingma2013auto}
Kingma, D.~P.
\newblock Auto-encoding variational bayes.
\newblock \emph{arXiv preprint arXiv:1312.6114}, 2013.

\bibitem[Korba et~al.(2021)Korba, Aubin-Frankowski, Majewski, and Ablin]{korba2021kernel}
Korba, A., Aubin-Frankowski, P.-C., Majewski, S., and Ablin, P.
\newblock Kernel stein discrepancy descent.
\newblock In \emph{International Conference on Machine Learning}, pp.\  5719--5730. PMLR, 2021.

\bibitem[Krizhevsky et~al.(2009)Krizhevsky, Hinton, et~al.]{krizhevsky2009learning}
Krizhevsky, A., Hinton, G., et~al.
\newblock Learning multiple layers of features from tiny images.
\newblock 2009.

\bibitem[Kullback \& Leibler(1951)Kullback and Leibler]{kullback1951information}
Kullback, S. and Leibler, R.~A.
\newblock On information and sufficiency.
\newblock \emph{The annals of mathematical statistics}, 22\penalty0 (1):\penalty0 79--86, 1951.

\bibitem[Laguna et~al.(2019)Laguna, Salazar, and Sagar]{laguna2019entropic}
Laguna, H.~G., Salazar, S.~J., and Sagar, R.~P.
\newblock Entropic kullback-leibler type distance measures for quantum distributions.
\newblock \emph{International Journal of Quantum Chemistry}, 119\penalty0 (19):\penalty0 e25984, 2019.

\bibitem[Li et~al.(2024)Li, Dulal, Ohorodnikov, Wang, and Ding]{li2024efficient}
Li, D., Dulal, D., Ohorodnikov, M., Wang, H., and Ding, Y.
\newblock Efficient quantum gradient and higher-order derivative estimation via generalized hadamard test.
\newblock \emph{arXiv preprint arXiv:2408.05406}, 2024.

\bibitem[liang et~al.(2021)liang, Wu, Li, Wang, Meng, Qin, Chen, Zhang, and Liu]{wu2021rdrop}
liang, x., Wu, L., Li, J., Wang, Y., Meng, Q., Qin, T., Chen, W., Zhang, M., and Liu, T.-Y.
\newblock R-drop: Regularized dropout for neural networks.
\newblock In Ranzato, M., Beygelzimer, A., Dauphin, Y., Liang, P., and Vaughan, J.~W. (eds.), \emph{Advances in Neural Information Processing Systems}, volume~34, pp.\  10890--10905. Curran Associates, Inc., 2021.

\bibitem[Liang et~al.(2022)Liang, Wang, Cheng, Ding, Ren, Gao, Hu, Boning, Qian, Han, et~al.]{liang2022variational}
Liang, Z., Wang, H., Cheng, J., Ding, Y., Ren, H., Gao, Z., Hu, Z., Boning, D.~S., Qian, X., Han, S., et~al.
\newblock Variational quantum pulse learning.
\newblock In \emph{2022 IEEE International Conference on Quantum Computing and Engineering (QCE)}, pp.\  556--565. IEEE, 2022.

\bibitem[Liang et~al.(2024)Liang, Cheng, Ren, Wang, Hua, Song, Ding, Chong, Han, Qian, et~al.]{liang2024napa}
Liang, Z., Cheng, J., Ren, H., Wang, H., Hua, F., Song, Z., Ding, Y., Chong, F.~T., Han, S., Qian, X., et~al.
\newblock Napa: intermediate-level variational native-pulse ansatz for variational quantum algorithms.
\newblock \emph{IEEE Transactions on Computer-Aided Design of Integrated Circuits and Systems}, 2024.

\bibitem[Lin(1991)]{lin1991divergence}
Lin, J.
\newblock Divergence measures based on the shannon entropy.
\newblock \emph{IEEE Transactions on Information theory}, 37\penalty0 (1):\penalty0 145--151, 1991.

\bibitem[Liu \& Wang(2016)Liu and Wang]{liu2016stein}
Liu, Q. and Wang, D.
\newblock Stein variational gradient descent: A general purpose bayesian inference algorithm.
\newblock \emph{Advances in neural information processing systems}, 29, 2016.

\bibitem[Liu(2019)]{liu2019roberta}
Liu, Y.
\newblock Roberta: A robustly optimized bert pretraining approach.
\newblock \emph{arXiv preprint arXiv:1907.11692}, 364, 2019.

\bibitem[MacKay(2003)]{mackay2003information}
MacKay, D.~J.
\newblock \emph{Information theory, inference and learning algorithms}.
\newblock Cambridge university press, 2003.

\bibitem[Men{\'e}ndez et~al.(1997)Men{\'e}ndez, Pardo, Pardo, and Pardo]{menendez1997jensen}
Men{\'e}ndez, M.~L., Pardo, J., Pardo, L., and Pardo, M.
\newblock The jensen-shannon divergence.
\newblock \emph{Journal of the Franklin Institute}, 334\penalty0 (2):\penalty0 307--318, 1997.

\bibitem[Moreno et~al.(2003)Moreno, Ho, and Vasconcelos]{moreno2003kullback}
Moreno, P., Ho, P., and Vasconcelos, N.
\newblock A kullback-leibler divergence based kernel for svm classification in multimedia applications.
\newblock \emph{Advances in neural information processing systems}, 16, 2003.

\bibitem[Mottonen et~al.(2004)Mottonen, Vartiainen, Bergholm, and Salomaa]{mottonen2004transformation}
Mottonen, M., Vartiainen, J.~J., Bergholm, V., and Salomaa, M.~M.
\newblock Transformation of quantum states using uniformly controlled rotations.
\newblock \emph{arXiv preprint quant-ph/0407010}, 2004.

\bibitem[Nielsen \& Chuang(2010)Nielsen and Chuang]{nielsen2010quantum}
Nielsen, M.~A. and Chuang, I.~L.
\newblock \emph{Quantum computation and quantum information}.
\newblock Cambridge university press, 2010.

\bibitem[Peng et~al.(2024)Peng, Li, Liang, and Wang]{peng2024qsco}
Peng, Y., Li, X., Liang, Z., and Wang, Y.
\newblock Qsco: A quantum scoring module for open-set supervised anomaly detection.
\newblock \emph{arXiv preprint arXiv:2405.16368}, 2024.

\bibitem[Pham \& Le(2021)Pham and Le]{pham2021autodropout}
Pham, H. and Le, Q.
\newblock Autodropout: Learning dropout patterns to regularize deep networks.
\newblock In \emph{Proceedings of the AAAI conference on artificial intelligence}, volume~35, pp.\  9351--9359, 2021.

\bibitem[Preskill(2018)]{preskill2018quantum}
Preskill, J.
\newblock Quantum computing in the nisq era and beyond.
\newblock \emph{Quantum}, 2:\penalty0 79, 2018.

\bibitem[Robbins \& Monro(1951)Robbins and Monro]{robbins1951stochastic}
Robbins, H. and Monro, S.
\newblock A stochastic approximation method.
\newblock \emph{The annals of mathematical statistics}, pp.\  400--407, 1951.

\bibitem[Roberts \& Rosenthal(2004)Roberts and Rosenthal]{roberts2004general}
Roberts, G.~O. and Rosenthal, J.~S.
\newblock General state space markov chains and mcmc algorithms.
\newblock 2004.

\bibitem[Shannon(1948)]{shannon1948mathematical}
Shannon, C.~E.
\newblock A mathematical theory of communications.
\newblock \emph{Bell system technical journal}, 27:\penalty0 379--423, 1948.

\bibitem[Srivastava et~al.(2014)Srivastava, Hinton, Krizhevsky, Sutskever, and Salakhutdinov]{srivastava2014dropout}
Srivastava, N., Hinton, G., Krizhevsky, A., Sutskever, I., and Salakhutdinov, R.
\newblock Dropout: a simple way to prevent neural networks from overfitting.
\newblock \emph{The journal of machine learning research}, 15\penalty0 (1):\penalty0 1929--1958, 2014.

\bibitem[Tompson et~al.(2015)Tompson, Goroshin, Jain, LeCun, and Bregler]{tompson2015efficientspatialdrop}
Tompson, J., Goroshin, R., Jain, A., LeCun, Y., and Bregler, C.
\newblock Efficient object localization using convolutional networks.
\newblock In \emph{Proceedings of the IEEE conference on computer vision and pattern recognition}, pp.\  648--656, 2015.

\bibitem[Van~Erven \& Harremos(2014)Van~Erven and Harremos]{van2014renyi}
Van~Erven, T. and Harremos, P.
\newblock R{\'e}nyi divergence and kullback-leibler divergence.
\newblock \emph{IEEE Transactions on Information Theory}, 60\penalty0 (7):\penalty0 3797--3820, 2014.

\bibitem[Vedral(2002)]{vedral2002role}
Vedral, V.
\newblock The role of relative entropy in quantum information theory.
\newblock \emph{Reviews of Modern Physics}, 74\penalty0 (1):\penalty0 197, 2002.

\bibitem[Von~Neumann(2013)]{von2013mathematische}
Von~Neumann, J.
\newblock \emph{Mathematische grundlagen der quantenmechanik}, volume~38.
\newblock Springer-Verlag, 2013.

\bibitem[Wager et~al.(2013)Wager, Wang, and Liang]{wager2013dropout}
Wager, S., Wang, S., and Liang, P.~S.
\newblock Dropout training as adaptive regularization.
\newblock \emph{Advances in neural information processing systems}, 26, 2013.

\bibitem[Wang(2018)]{wang2018glue}
Wang, A.
\newblock Glue: A multi-task benchmark and analysis platform for natural language understanding.
\newblock \emph{arXiv preprint arXiv:1804.07461}, 2018.

\bibitem[Wang et~al.(2022)Wang, Ding, Gu, Lin, Pan, Chong, and Han]{wang2022quantumnas}
Wang, H., Ding, Y., Gu, J., Lin, Y., Pan, D.~Z., Chong, F.~T., and Han, S.
\newblock Quantumnas: Noise-adaptive search for robust quantum circuits.
\newblock In \emph{2022 IEEE International Symposium on High-Performance Computer Architecture (HPCA)}, pp.\  692--708. IEEE, 2022.

\bibitem[Wilde(2013)]{wilde2013quantum}
Wilde, M.~M.
\newblock \emph{Quantum information theory}.
\newblock Cambridge university press, 2013.

\end{thebibliography}
\bibliographystyle{icml2025}

%%%%%%%%%%%%%%%%%%%%%%%%%%%%%%%%%%%%%%%%%%%%%%%%%%%%%%%%%%%%%%%%%%%%%%%%%%%%%%%
%%%%%%%%%%%%%%%%%%%%%%%%%%%%%%%%%%%%%%%%%%%%%%%%%%%%%%%%%%%%%%%%%%%%%%%%%%%%%%%
% APPENDIX
%%%%%%%%%%%%%%%%%%%%%%%%%%%%%%%%%%%%%%%%%%%%%%%%%%%%%%%%%%%%%%%%%%%%%%%%%%%%%%%
%%%%%%%%%%%%%%%%%%%%%%%%%%%%%%%%%%%%%%%%%%%%%%%%%%%%%%%%%%%%%%%%%%%%%%%%%%%%%%%
\newpage
\appendix
\onecolumn

% You can have as much text here as you want. The main body must be at most $8$ pages long.
% For the final version, one more page can be added.
% If you want, you can use an appendix like this one.  

% The $\mathtt{\backslash onecolumn}$ command above can be kept in place if you prefer a one-column appendix, or can be removed if you prefer a two-column appendix.  Apart from this possible change, the style (font size, spacing, margins, page numbering, etc.) should be kept the same as the main body.
% %%%%%%%%%%%%%%%%%%%%%%%%%%%%%%%%%%%%%%%%%%%%%%%%%%%%%%%%%%%%%%%%%%%%%%%%%%%%%%%
% %%%%%%%%%%%%%%%%%%%%%%%%%%%%%%%%%%%%%%%%%%%%%%%%%%%%%%%%%%%%%%%%%%%%%%%%%%%%%%%

\end{document}